\documentclass[journal,twoside,web]{ieeecolor}
\usepackage{generic}

\usepackage{amsmath,amssymb,amsfonts}
\usepackage{algorithmic}
\usepackage{graphicx}
\usepackage{textcomp}
\usepackage{xcolor}
\usepackage{epstopdf}
\usepackage{mathtools}
\usepackage{framed}

\usepackage{cite}
\usepackage[
    colorlinks=true,
    linkcolor=blue,
    citecolor=blue,
    urlcolor=blue
]{hyperref}

\makeatletter
\let\NAT@parse\undefined
\makeatother

\usepackage[
    capitalize,      
]{cleveref}
\crefname{assumption}{Assumption}{Assumptions}

\newtheorem{theorem}{Theorem}
\newtheorem{lemma}{Lemma}

\newtheorem{remark}{Remark}
\newtheorem{assumption}{Assumption}

\newtheorem{definition}{Definition}

\def\BibTeX{{\rm B\kern-.05em{\sc i\kern-.025em b}\kern-.08em
    T\kern-.1667em\lower.7ex\hbox{E}\kern-.125emX}}

\begin{document}

\title{Geometric Fixed-Time Sliding Mode Control for Constrained  Attitude Tracking on $\mathrm{SO}(3)$\\
}



\author{Saumitra Barman, Shashi Ranjan Kumar, Senior Member, IEEE,  and Rohit Gupta
\thanks{The authors are with the Department of Aerospace Engineering, Indian Institute of Technology Bombay, Powai, Mumbai 400076, India (e-mails: \{saumitra, srk, rohit\}@aero.iitb.ac.in).}
}

\maketitle





\begin{abstract}
This paper studies constrained spacecraft attitude tracking on the Riemannian configuration manifold $\mathrm{SO}(3)$ in the presence of multiple attitude pointing constraints and matched external disturbances. To address this, an attitude potential function is proposed intrinsically on $\mathrm{SO}(3)$, and its key properties are established using intrinsic geometric analysis. Under mild conditions, the potential function is shown to admit a unique nondegenerate minimum at the desired attitude over the admissible subset of $\mathrm{SO}(3)$, defined by excluding the forbidden attitude regions as well as a measure-zero set, thereby ensuring a well-posed constrained attitude tracking problem. A Riemannian Hessian analysis shows that the Hessian of the potential function is locally uniform positive definite in an open neighborhood of the desired attitude, thereby establishing local strong convexity. A nonsingular fixed-time geometric sliding manifold is proposed using the Riemannian gradient of the potential function, leading to a geometric fixed-time sliding-mode-based constrained attitude control law. It is shown that, for every initial attitude in the admissible subset, the closed-loop state trajectory evolves on $\mathrm{SO}(3)\times\mathbb{R}^3$, with the attitude remaining in the admissible subset throughout the maneuver, while the state converges to a sufficiently small compact neighborhood of the desired equilibrium in a prescribed fixed time. Numerical simulations validate the proposed control approach and illustrate the theoretical results.
\end{abstract}


\section{Introduction}
Spacecraft attitude tracking requires steering a spacecraft from its current orientation to a desired orientation in three-dimensional space, while ensuring that a body-fixed instrument boresight does not point towards undesirable directions, for example, keeping  telescopic instruments away from bright celestial objects to avoid damage and preserve mission functionality. In this setting, the control objective is to develop a guidance and control strategy that drives the spacecraft to the commanded attitude while consistently avoiding these prohibited pointing directions. A prominent example is the Cassini spacecraft, in which a dedicated constraint-monitoring mechanism was employed in \cite{singh1997constraint} to prevent specific onboard sensors from being directed  towards the Sun.

To address this challenge, one class of methods focuses on trajectory planning, where a feasible attitude path is generated in advance while explicitly incorporating pointing constraints. A variety of attitude trajectory planning strategies have been proposed in the literature, including constraint-monitoring schemes (see, e.g., \cite{singh1997constraint}), geometric planners (see, e.g., \cite{hablani1999attitude,biggs2016geometric}), randomized exploration techniques (see, e.g., \cite{feron2012randomized}), discretization-based approaches (see, e.g, \cite{kjellberg2013discretized,tanygin2017fast}), recursive methods that decompose rotational paths (see, e.g., \cite{xu2018rotational}), gradient-driven optimization planners (see, e.g., \cite{celani2020spacecraft}). However, such approaches generally demand substantial computational resources, resulting in high costs for real-time onboard implementation. A second major line of research incorporates constraints into the control design through artificial potential field (APF)–based formulations, in which the potential function attains a minimum at the target reference and takes large values in the vicinity of exclusion zones. Within the APF-based framework, several control architectures have been developed, including proportional–derivative controllers (see, e.g., \cite{mclnnes1994large,shen2017velocity,nicotra2019spacecraft}), backstepping designs (see, e.g., \cite{lee2014feedback,cheng2018spacecraft}), sliding-mode schemes (see, e.g., \cite{shen2018rigid,yang2021potential,li2023optimal}), and kinematic steering laws (see, e.g., \cite{diaz2018kinematic}). Although the offline synthesis of such controllers can be relatively involved, their resulting explicit closed-form state-feedback structure leads to modest computational effort, making them well-suited for real-time implementation.

Sliding mode control (SMC) has been extensively employed in spacecraft attitude control owing to its inherent robustness  property to model uncertainties, nonlinearities, and bounded exogenous disturbances (see, e.g., \cite{utkin2003variable,bhat2002continuous}). Despite these advantages, standard SMC formulations do not explicitly account for attitude constraints, and constraint satisfaction cannot, in general, be guaranteed. To overcome this limitation, several works have combined SMC with artificial potential function (APF) techniques to incorporate constraint handling (see, e.g., \cite{shen2018rigid,yang2021potential,li2023optimal}). However, such approaches typically guarantee only asymptotic convergence of the closed-loop system, for which the settling time is theoretically infinite. This limitation may be undesirable in safety-critical spacecraft attitude maneuvers that require guaranteed convergence within a fixed amount of time. Motivated by this limitation, the fixed-time stability concept originally introduced in~\cite{polyakov2011nonlinear} was employed in~\cite{li2023optimal} to design a fixed-time sliding mode control scheme for constrained spacecraft attitude control.

Nevertheless, all the above-discussed APF-based spacecraft attitude controllers rely on Euclidean-space attitude parameterizations, such as Euler angles (see, e.g., \cite{mclnnes1994large}) or modified Rodrigues parameters (see, e.g., \cite{diaz2018kinematic}). Although convenient, these parameterizations are  valid only locally and inevitably suffer from coordinate singularities, which can lead to undesirable numerical and control issues near the singular configurations (see, e.g., \cite{schaub1996stereographic}).  On the other hand, unit quaternions provide a globally nonsingular alternative and have therefore been widely employed in constructing potential functions (see, e.g., \cite{shen2017velocity,nicotra2019spacecraft,lee2014feedback,cheng2018spacecraft,yang2021potential}); however, they introduce an inherent ambiguity, since the set of unit quaternions $\mathrm{S}^{3}$, is a double cover of $\mathrm{SO}(3)$ (see, e.g., \cite{mayhew2011quaternion}). As a result, a pair of antipodal unit quaternions represents the same physical attitude, which can lead to inconsistent control actions and the well-known unwinding phenomenon, whereby the controller unnecessarily commands a large rotation before reaching the desired attitude for certain initial conditions (see, e.g., \cite{bhat2000topological}). In contrast, attitude tracking directly on the underlying configuration manifold $\mathrm{SO}(3)$ have been shown to achieve almost-global stabilization without encountering coordinate singularities or the unwinding phenomenon (see, e.g., \cite{chaturvedi2011rigid}). Controllers formulated intrinsically based on the configuration manifold, commonly referred to as geometric controllers, have been studied extensively in the literature (see, e.g., \cite{lee2012exponential,bullo1999tracking,maithripala2015intrinsic,lee2015global,maithripala2006almost}). By avoiding local parameterizations, geometric control approaches exploit the underlying structure of the configuration manifold, revealing properties that are not captured by coordinate-based formulations. However, these geometric controllers typically guarantee only asymptotic convergence to the desired attitude and are not explicitly designed to accommodate system uncertainties or external disturbances. To address these limitations, fixed-time geometric sliding mode control strategies have been proposed (see, e.g., \cite{shi2017finite,barman2024almost}).

Despite extensive work on unconstrained attitude tracking, constrained attitude tracking on $\mathrm{SO}(3)$ has received comparatively limited attention. Specifically, a barrier function-based controller was proposed in \cite{kulumani2017constrained}, a  path planning-based approach was developed in \cite{tan2020constrained}, and a model predictive control formulation  was presented in \cite{lee2017geometric}. The path planning-based method in \cite{tan2020constrained} generates feasible constrained trajectories through sampling or rotation-space search; however, it requires offline path computation and does not inherently guarantee closed-loop invariance. The model predictive control approach in \cite{lee2017geometric} enforces constraints via online optimization, offering flexibility at the expense of increased computational complexity and without guarantees of finite-time convergence. The barrier function-based geometric controller in \cite{kulumani2017constrained} enforces attitude constraint avoidance on $\mathrm{SO}(3)\times \mathbb{R}^3$. However, this approach may introduce multiple critical points within the admissible free space, i.e., the set of attitude configurations that satisfy the prescribed pointing constraints, without guaranteeing their nondegeneracy, which can cause the system to converge to undesired configurations. Moreover, it typically guarantees only asymptotic convergence and does not explicitly characterize the maximal invariant admissible set. Furthermore, to the best of the authors’ knowledge, sliding mode control–based constrained spacecraft attitude tracking on $\mathrm{SO}(3)$ has not yet been reported in the literature.

Motivated by these observations, this paper develops a geometric sliding-mode controller on $\mathrm{SO}(3)\times \mathbb{R}^3$ that integrates a proposed attitude potential function with a nonsingular fixed-time sliding manifold for constrained spacecraft attitude maneuvering.  The proposed controller renders the admissible free space, excluding a zero-measure set, the largest positively invariant set of the closed-loop dynamics, thereby ensuring multiple pointing constraints throughout all admissible attitude maneuvers. It requires no offline or online optimization, yielding a computationally efficient closed-loop solution. Finally, the effectiveness of the proposed approach is demonstrated through two numerical examples. The main contributions of this work are summarized as follows:

\begin{enumerate}
\item[(i)] An attitude potential function is intrinsically constructed on the Riemannian manifold $\mathrm{SO}(3)$ for constrained spacecraft attitude maneuvering, and its key properties are established directly on the manifold rather than through an extrinsic formulation in the higher-dimensional Euclidean space. Under mild conditions, the potential function is shown to admit a unique nondegenerate minimum at the desired attitude over the admissible subset of $\mathrm{SO}(3)$, defined by excluding the forbidden attitude regions as well as a measure-zero set, thereby ensuring a well-posed constrained attitude tracking problem. Moreover, a Riemannian Hessian analysis shows that the Hessian of the potential function is locally uniform positive definite in an open neighborhood of the desired attitude, thereby establishing local strong convexity.

\item[(ii)] A nonsingular fixed-time sliding manifold is proposed on $\mathrm{SO}(3)\times \mathbb{R}^3$ by exploiting the Riemannian gradient of the attitude potential function. Based on this manifold, a geometric fixed-time sliding-mode control law is developed for constrained spacecraft attitude tracking in the presence of multiple pointing constraints and matched external disturbances. To the best of the authors' knowledge, such a sliding-mode formulation posed directly on  $\mathrm{SO}(3)\times \mathbb{R}^3$ has not been previously reported. Moreover, the equilibrium corresponding to the desired attitude is shown to be unique for the resulting closed-loop attitude dynamics. Furthermore, it is shown that, for every initial attitude in the admissible subset, the closed-loop state trajectory evolves on $\mathrm{SO}(3)\times\mathbb{R}^3$, with the attitude remaining in the admissible subset throughout the maneuver, while the state converges to a sufficiently small compact neighborhood of the desired equilibrium in a prescribed fixed time.
\end{enumerate}


\section{Preliminaries and Problem Statement}
\subsection{Notations}
The following mathematical notations are adopted throughout the paper. Let a Riemannian manifold $\mathcal{M}$ be endowed with a Riemannian metric $\langle\cdot,\cdot\rangle_{\boldsymbol{x}}$ on each tangent space $T_{\boldsymbol{x}}\mathcal{M}$, where $\boldsymbol{x}\in\mathcal{M}$. The tangent bundle of $\mathcal{M}$ is denoted by $T\mathcal{M}$. The notation $\operatorname{int}(\mathcal{M})$ denotes the topological interior of $\mathcal{M}$, while $\partial\mathcal{M}$ and $\partial_{\operatorname{top}}\mathcal{M}$ denote the manifold boundary and topological boundary of $\mathcal{M}$, respectively. The notation $\nabla_{\boldsymbol{\xi}}$ denotes the covariant derivative induced by the Levi--Civita connection associated with the underlying Riemannian metric, evaluated along the vector field $\boldsymbol{\xi}\in T_{\boldsymbol{x}}\mathcal{M}$. The Riemannian volume measure on $\mathcal{M}$ is denoted by $\mu_{\mathcal{M}}$.  The $3\times3$ identity matrix is defined as $\mathbb{I}_{3\times3}\in\mathbb{R}^{3\times3}$. For $\boldsymbol{x},\boldsymbol{y}\in\mathbb{R}^{3}$, their inner product is denoted by $\boldsymbol{x}\cdot\boldsymbol{y}\in\mathbb{R}$. The unit sphere is defined as $\mathbb{S}^2=\{ \boldsymbol{a}\in\mathbb{R}^{3} : \|\boldsymbol{a}\|=1 \}$. The operator $\operatorname{tr}(\cdot)$ denotes the matrix trace, and $\|\cdot\|$ denotes the Euclidean norm for vectors and the Frobenius norm for matrices. The special orthogonal group is defined as $\mathrm{SO}(3)=\{ \boldsymbol{R}\in\mathbb{R}^{3\times3} : \boldsymbol{R}^\top\boldsymbol{R}=\mathbb{I}_{3\times3},\ \det\boldsymbol{R}=1 \}$, and its Lie algebra is defined as $\mathfrak{so}(3)=\{ \boldsymbol{S}\in\mathbb{R}^{3\times3} : \boldsymbol{S}^\top=-\boldsymbol{S} \}$, where $\det(\cdot)$ denotes the determinant operator.  The tangent space of $\mathrm{SO}(3)$ at $\boldsymbol{R}$ is given by $T_{\boldsymbol{R}}\mathrm{SO}(3)=\{ \boldsymbol{R}\boldsymbol{X}\in\mathbb{R}^{3\times3} : \boldsymbol{X}^{\top}=-\boldsymbol{X} \}$. The bi-invariant Riemannian metric on $\mathrm{SO}(3)$ is defined as 
$\langle\boldsymbol{M},\boldsymbol{N}\rangle_{\boldsymbol{R}}
=\operatorname{tr}(\boldsymbol{M}^{\top}\boldsymbol{N})$, where
$\boldsymbol{R}\in \mathrm{SO}(3)$ and
$\boldsymbol{M},\boldsymbol{N}\in T_{\boldsymbol{R}}\mathrm{SO}(3)$. The cross map $(\cdot)^{\times}:\mathbb{R}^{3}\rightarrow\mathfrak{so}(3)$
is defined by $\boldsymbol{x}^{\times}\coloneq[0\;\;-x_{3}\;\;x_{2};\;\;x_{3}\;\;0\;\;x_{1};\;\;-x_{2}\;\;-x_{1}\;\;0]$, where $\boldsymbol{x} \coloneq [x_{1}~x_{2}~x_{3}]^{\top}\in\mathbb{R}^{3}$, and the ``vee'' map $(\cdot)^\vee:\mathfrak{so}(3)\to\mathbb{R}^{3}$ is its inverse. For a subset $\mathcal{S}\subset\mathcal{M}$, $\overline{\mathcal{S}}$ denotes its closure in $\mathcal{M}$. The skew-symmetric part of a matrix $\boldsymbol{A}$ is denoted by $\operatorname{skew}(\boldsymbol{A})\coloneq(\boldsymbol{A}-\boldsymbol{A}^{\top})/2$. For a smooth function $\Psi:\mathrm{SO}(3)\to\mathbb{R}$, $\operatorname{grad}\Psi$ and $\operatorname{Hess}\Psi$ denote the
Riemannian gradient and Hessian, while $\nabla\Psi$ and
$\nabla^{2}\Psi$ denote their Euclidean counterparts. The notation $\lambda_{\max}(\boldsymbol{A})$ gives the largest eigenvalue of a symmetric matrix $\boldsymbol{A}$. The operators $\max(\cdot)$ and $\min(\cdot)$ return the maximum and minimum of their arguments, respectively. The notion $[x]_+$ is defined as $[x]_+ \coloneq \max\{x,0\}$, where $x\in\mathbb{R}$. The time index is denoted by $t\in \mathbb{R}_{\geq0}$. The sets of nonnegative and positive real numbers are denoted by $\mathbb{R}_{\geq 0}\coloneqq\{x\in\mathbb{R}:x\geq 0\}$ and $\mathbb{R}_{>0}\coloneqq\{x\in\mathbb{R}:x>0\}$, respectively.

\subsection{Definitions and Lemmas}
The following definitions and fixed-time stability condition extend the framework of \cite{polyakov2011nonlinear} to dynamical systems on Riemannian manifolds. Consider the following autonomous dynamical system on a Riemannian manifold $\mathcal{M}$:
\begin{equation}
\dot{\boldsymbol{x}}(t)=\boldsymbol{f}(\boldsymbol{x}(t)),
\qquad
\boldsymbol{x}(0)=\boldsymbol{x}_0,
\label{eq:autonomous_eq}
\end{equation}
where $\boldsymbol{x}(t)\in\mathcal{M}$ and
$\boldsymbol{f}(\boldsymbol{x})\in T_{\boldsymbol{x}}\mathcal{M}$.
The vector field $\boldsymbol{f}$ may be discontinuous, and the
solutions of \eqref{eq:autonomous_eq} are understood in the sense of
Filippov (see, e.g., \cite{filippov1988}).
Furthermore, $\boldsymbol{x}_d\in\mathcal{M}$ denotes the desired
equilibrium of \eqref{eq:autonomous_eq}.

\begin{definition}
The desired equilibrium $\boldsymbol{x}_d$ of \eqref{eq:autonomous_eq} is said to be almost-globally finite-time stable on $\mathcal{M}$ if there exists a subset $\mathcal{M}_0\subseteq\mathcal{M}$ containing $\boldsymbol{x}_d$ such that $\mu_{\mathcal{M}}(\mathcal{M}\setminus\mathcal{M}_0)=0$, $\boldsymbol{x}_d$ is asymptotically stable relative to $\mathcal{M}_0$, and every solution $\boldsymbol{x}(\cdot,\boldsymbol{x}_0)$ with $\boldsymbol{x}_0\in\mathcal{M}_0$ reaches $\boldsymbol{x}_d$ in finite time. That is, there exists a settling-time function $T:\mathcal{M}_0\rightarrow\mathbb{R}_{\geq0}$ such that $\boldsymbol{x}(t,\boldsymbol{x}_0)=\boldsymbol{x}_d$ for all  $t\geq T(\boldsymbol{x}_0)$ and $\boldsymbol{x}_0\in\mathcal{M}_0$.
\end{definition}

\begin{definition}
The equilibrium $\boldsymbol{x}_d$ of \eqref{eq:autonomous_eq} is said to be almost-globally fixed-time stable on $\mathcal{M}$ if it is almost-globally finite-time stable and there exists a constant $T_{\max}\in\mathbb{R}_{>0}$ such that $T(\boldsymbol{x}_0)\leq T_{\max}$ for all $\boldsymbol{x}_0\in\mathcal{M}_0$.
\end{definition}
\begin{definition}
Consider the autonomous dynamical system \eqref{eq:autonomous_eq}. A set $\mathcal{E}\subset\mathcal{M}$ is said to be almost-globally fixed-time attractive on
$\mathcal{M}$ if there exists a subset
$\mathcal{M}_0\subseteq\mathcal{M}$ satisfying
$\mu_{\mathcal{M}}(\mathcal{M}\setminus\mathcal{M}_0)=0$ and a constant
$T_{\max}\in\mathbb{R}_{>0}$ such that, for every initial condition
$\boldsymbol{x}_0\in\mathcal{M}_0$, the corresponding solution satisfies
$\boldsymbol{x}(t,\boldsymbol{x}_0)\in\mathcal{E}$ for all $t\geq T(\boldsymbol{x}_0)$, where
$T(\boldsymbol{x}_0)\leq T_{\max}$ for all $\boldsymbol{x}_0\in\mathcal{M}_0$. In addition,
$\mathcal{E}$ is forward invariant, i.e., $\boldsymbol{x}_0\in\mathcal{E}$ implies
$\boldsymbol{x}(t)\in\mathcal{E}$ for all $t \in \mathbb{R}_{\geq0}$.
\end{definition}

\begin{lemma}[Manifold extension of fixed-time criterion
\cite{polyakov2011nonlinear}]
\label{lem:1}
Consider the autonomous dynamical system \eqref{eq:autonomous_eq}. Suppose there exists a subset $\mathcal{M}_0\subseteq\mathcal{M}$ such that $\mu_{\mathcal{M}}(\mathcal{M}\setminus\mathcal{M}_0)=0$, and let $\mathcal{E}\subset\mathcal{M}_0$ be a closed set such that $\boldsymbol{x}_d\in \operatorname{int}(\mathcal{E})$. Suppose there exists a continuously differentiable nonnegative function $V:\mathcal{M}_0\rightarrow\mathbb{R}_{\geq0}$ such that $V(\boldsymbol{x}_d)=0$ and $V(\boldsymbol{x})>0$ for all
$\boldsymbol{x}\in\mathcal{M}_0\setminus\{\boldsymbol{x}_d\}$ and a constant $0<\epsilon<1$ such that $\mathcal{E}=\{\boldsymbol{x}\in\mathcal{M}_0:V(\boldsymbol{x})\leq\epsilon\}$. Moreover, suppose that, along every trajectory of \eqref{eq:autonomous_eq} with $\boldsymbol{x}_0\in\mathcal{M}_0$, the function $V$ satisfies $\dot{V}(\boldsymbol{x})\leq-k_1V^p(\boldsymbol{x})-k_2V^q(\boldsymbol{x})$ for all $\boldsymbol{x}\in\mathcal{M}_0\setminus\mathcal{E}$ and $\dot{V}(\boldsymbol{x})<0$ for all $\boldsymbol{x}\in\mathcal{E}\setminus\{\boldsymbol{x}_d\}$, where $k_1,k_2\in\mathbb{R}_{>0}$ and $0<p<1<q$. Then, $\mathcal{E}$ is almost-globally fixed-time attractive and forward invariant on $\mathcal{M}$, and $T(\boldsymbol{x}_0)\leq1/(k_1(1-p))+(k_2(q-1))$ for all $\boldsymbol{x}_0\in\mathcal{M}_0$.
\end{lemma}

\begin{remark}
When the target set reduces to the singleton
$\mathcal{E}=\{\boldsymbol{x}_d\}$ and
$\dot{V}\leq-k_1V^p-k_2V^q$ for all
$\boldsymbol{x}\in\mathcal{M}_0\setminus\{\boldsymbol{x}_d\}$, \cref{lem:1} reduces to the standard almost-global fixed-time stability criterion for the desired equilibrium $\boldsymbol{x}_d$.
\end{remark}

\subsection{Spacecraft Attitude Kinematics and Dynamics}
Let $\mathcal{F}_{\mathcal{B}}$ denote a body-fixed frame attached to the spacecraft, with its origin located at the spacecraft’s center of mass. Similarly, let $\mathcal{F}_{\mathcal{N}}$ denote the inertial reference frame with respect to which the spacecraft dynamics is described. The spacecraft attitude with respect to the $\mathcal{F}_{\mathcal{N}}$ frame is represented by a rotation matrix $\boldsymbol{R} \in \mathrm{SO}(3)$. Let $\boldsymbol{\omega} = [\omega_{1}~\omega_{2}~\omega_{3}]^{\top}\in\mathbb{R}^{3}$ denote the angular velocity vector of the spacecraft with respect to the $\mathcal{F}_{\mathcal{N}}$ frame, with components expressed in the $\mathcal{F}_{\mathcal{B}}$ frame. The kinematic equation of the spacecraft in terms of rotation matrix is given by $\dot{\boldsymbol{R}}=\boldsymbol{R}\boldsymbol{\omega}^{\times}$ (see, e.g., \cite{chaturvedi2011rigid}). To formulate the attitude tracking objective on $\mathrm{SO}(3)$, a desired frame $\mathcal{F}_{\mathcal{D}}$ is introduced as the reference configuration that the spacecraft body-fixed frame $\mathcal{F}_{\mathcal{B}}$ is required to track. The attitude error between $\mathcal{F}_{\mathcal{B}}$ and $\mathcal{F}_{\mathcal{D}}$ is represented by the error rotation matrix $\boldsymbol{R}_{e}\in\mathrm{SO}(3)$, defined as $\boldsymbol{R}_{e}=\boldsymbol{R}_{d}^{\top}\boldsymbol{R}$, where $\boldsymbol{R}_{d}\in\mathrm{SO}(3)$ denotes the desired attitude of the spacecraft with respect to the inertial frame $\mathcal{F}_{\mathcal{N}}$. The angular velocity error vector $\boldsymbol{e}_{\omega}:(\mathrm{SO}(3))^{2}\times(\mathbb{R}^{3})^{2}\rightarrow\mathbb{R}^{3}$ is defined as $\boldsymbol{e}_{\omega}(\boldsymbol{R},\boldsymbol{R}_{d},\boldsymbol{\omega},\boldsymbol{\omega}_{d})
\coloneq
\boldsymbol{\omega}-\boldsymbol{R}_{e}^{\top}\boldsymbol{\omega}_{d}$, where $\boldsymbol{\omega}_{d}\in\mathbb{R}^{3}$ is the desired angular velocity vector expressed in the desired frame $\mathcal{F}_{\mathcal{D}}$. The attitude error kinematics of the spacecraft is given by
\begin{equation}
\dot{\boldsymbol{R}}_{e}=\boldsymbol{R}_{e}\boldsymbol{e}_{\omega}^{\times}.
\label{eq:kinematics}
\end{equation}
The time derivative of $\boldsymbol{e}_\omega$ is given by $\dot{\boldsymbol{e}}_{\omega}
=
\dot{\boldsymbol{\omega}}
+\boldsymbol{e}_{\omega}^{\times}\boldsymbol{R}_{e}^{\top}\boldsymbol{\omega}_{d}
-\boldsymbol{R}_{e}^{\top}\dot{\boldsymbol{\omega}}_{d}$. Therefore, using Euler's equation of motion, the attitude error dynamics of the spacecraft are obtained as
\begin{equation}
\boldsymbol{I}\dot{\boldsymbol{e}}_{\omega}
=
-\boldsymbol{\omega}^{\times}\boldsymbol{I}\boldsymbol{\omega}
+\boldsymbol{I}\boldsymbol{e}_{\omega}^{\times}\boldsymbol{R}_{e}^{\top}\boldsymbol{\omega}_{d}
-\boldsymbol{I}\boldsymbol{R}_{e}^{\top}\dot{\boldsymbol{\omega}}_{d}
+\boldsymbol{\tau}
+\boldsymbol{\tau}_{d},
\label{eq:dynamics}
\end{equation}
where $\boldsymbol{I}=\boldsymbol{I}^{\top}\in\mathbb{R}^{3\times3}$ is the spacecraft inertia matrix, and $\boldsymbol{\tau},\boldsymbol{\tau}_{d}\in\mathbb{R}^{3}$ denote the control and disturbance torques, respectively, applied about the spacecraft center of mass and expressed in the  $\mathcal{F}_{\mathcal{B}}$ frame.
\begin{assumption}
The desired signals $\boldsymbol{\omega}_d(t)$ and
$\dot{\boldsymbol{\omega}}_d(t)$ are assumed to be locally
bounded for all $t\in \mathbb{R}_{\geq0}$, and the disturbance $\boldsymbol{\tau}_d$ is Lebesgue measurable and  satisfies $\operatorname{sup}_{t \in \mathbb{R}_{\ge 0}} \| \boldsymbol{\tau}_d(t) \|
\leq
\bar{\tau}_d
<\gamma$, where $\bar{\tau}_d>0$ is a known upper bound on the disturbance magnitude, and $\gamma>0$ is a prescribed constant. 
\label{assump:1}
\end{assumption}

\subsection{Attitude Pointing Constraints and Forbidden Regions}
The pointing constraints introduced at the beginning of the Introduction are illustrated using the rotation matrix $\boldsymbol{R}\in\mathrm{SO}(3)$. Let $\boldsymbol{g}_{d}\in\mathbb{S}^{2}$ denote direction of a sensitive axis (e.g., a sensor line-of-sight direction) of the spacecraft, expressed in  $\mathcal{F}_{\mathcal{B}}$ frame. Let $\boldsymbol{y}'_i\in\mathbb{S}^{2}$ represent $i^{\operatorname{th}}$ unit vector collinear with the forbidden direction in $\mathcal{F}_{\mathcal{N}}$ frame (e.g., a bright object like the Sun) that must not align too closely with $\boldsymbol{R}\boldsymbol{g}_{d}$ for
$i\in \mathcal{I}$, where $\mathcal{I}\coloneq\{1,...,n\}$ and $n$ denotes the total number of pointing constraints. Violating this constraint could result in damage to sensitive equipment or mission failure. To enforce this constraint,
we define the pointing constraint set $\mathcal{O}_{i}\subset\mathrm{SO}(3)$
for each forbidden direction $\boldsymbol{y}'_{i}$ as $\mathcal{O}_{i}\coloneq\left\{ \boldsymbol{R}\in\mathrm{SO}(3):\boldsymbol{y}{'}_{i}^{\top}\boldsymbol{R}\boldsymbol{g}_{d}>\cos\theta_{i},i\in \mathcal{I}\right\}$, where $\theta_{i}\in(0,\pi)$ is the minimum allowable angle between
the vector $\boldsymbol{R}\boldsymbol{g}_{d}$ and the forbidden inertial direction $\boldsymbol{y}'_{i}$.
Intuitively, the sensor line-of-sight direction $\boldsymbol{R}\boldsymbol{g}_{d}$ must remain outside the cone of half-angle $\theta_{i}$ centered about $\boldsymbol{y}_{i}'$ in $\mathcal{F}_{\mathcal{N}}$ frame. Moreover, we want all possible relative attitudes of the spacecraft
to exclude a subset defined by $\mathcal{L}\coloneq\{\boldsymbol{R}\in \mathrm{SO}(3): \operatorname{tr}(\boldsymbol{R}_d^{\top} \boldsymbol{R})>-1\}$. Therefore,   the admissible attitude set is defined as $\operatorname{int}(\mathcal{M})$, where 
\[
\mathcal{M}\coloneq \mathcal{L}\cap \left(\mathrm{SO}(3)\setminus\bigcup_{i=1}^{m}{\mathcal{O}}_{i}\right). 
\]
Moreover, the admissible attitude set in terms of  $\boldsymbol{R}_e$ can be expressed as $\operatorname{int}(\mathcal{M}_e)$, where $\mathcal{M}_e\coloneq\{\boldsymbol{R}_e\in\mathrm{SO}(3): 
\operatorname{tr}(\boldsymbol{R}_e)>-1,\ 
\boldsymbol{y}_{i}^{\top}\boldsymbol{R}_e\boldsymbol{g}_{d}\leq\cos\theta_i,\ i\in\mathcal{I}\}$,
where $\boldsymbol{y}_{i}=\boldsymbol{R}_d^{\top}\boldsymbol{y}'_{i}$.  It will be shown in the next section that $\mathcal{M}$ and $\mathcal{M}_{e}$ are smooth $3$-dimensional manifolds with boundary. To proceed, the following assumptions are needed.

\begin{assumption}
    The initial attitude satisfies $\boldsymbol{R}(0) \in \operatorname{int}(\mathcal{M})$ and the desired attitude satisfies $\boldsymbol{R}_d(t) \in \operatorname{int} (\mathcal{M}$) for all $t \in \mathbb{R}_{\ge 0}$.
    \label{assump:2}
\end{assumption}

\begin{assumption}
   The forbidden regions are pairwise disjoint, i.e., \(\overline{\mathcal{O}}_i\cap\overline{\mathcal{O}}_j=\emptyset\) for \(i\neq j\).
   \label{assump:3}
\end{assumption}
Assumption~2 ensures that both the initial and desired attitudes lie in the interior of the admissible attitude manifold $\mathcal{M}$.
\subsection{Problem Statement}
The control objective is to design a geometric fixed-time sliding mode control law on $\mathrm{SO}(3)\times\mathbb{R}^3$ such that the attitude constraints are satisfied throughout the maneuver, i.e., $\boldsymbol{R}(t)\in\operatorname{int}(\mathcal{M})$ for all $t\in\mathbb{R}_{\geq0}$, for every initial attitude
$\boldsymbol{R}(0)\in\operatorname{int}(\mathcal{M})$. Simultaneously, the control law is required to drive the error-state trajectory $(\boldsymbol{R}_e(t),\boldsymbol{e}_{\omega}(t))$ into a sufficiently small compact neighborhood $\mathcal{E}\subset
\operatorname{int}(\mathcal{M}_e)\times\mathbb{R}^3$ of the desired equilibrium $(\boldsymbol{\mathbb{I}}_{3\times3},\boldsymbol{0})$ within a prescribed fixed time $T_f\in\mathbb{R}_{>0}$. In particular, $(\boldsymbol{R}_e(t),\boldsymbol{e}_{\omega}(t))\in\mathcal{E}$ for all $t\geq T_f$, where
$(\boldsymbol{\mathbb{I}}_{3\times3},\boldsymbol{0})\in
\operatorname{int}(\mathcal{E})$. Consequently, the control law is required to render $\mathcal{E}$ almost-globally fixed-time attractive on $(\operatorname{int}(\mathcal{M}_e)\cup\mathcal{L})\times\mathbb{R}^3$ and forward invariant, while ensuring satisfaction of all attitude constraints throughout the maneuver.

\section{A Novel  Attitude Potential Function on $\mathrm{SO}(3)$}
To facilitate the construction and analysis of the proposed attitude potential function on $\mathrm{SO}(3)$, we first establish the geometric properties of the admissible attitude sets. Since the constraints are imposed on the actual attitude, whereas the tracking problem is formulated in terms of the attitude error, it is necessary to establish the correspondence between the admissible sets $\operatorname{int}(\mathcal{M})$ and $\operatorname{int}(\mathcal{M}_e)$. Furthermore, characterizing these sets as manifolds with boundary provides the geometric setting required to study scalar-valued functions defined on them, particularly their critical points and boundary behavior. These properties are essential for the development and analysis of the proposed potential function and will subsequently support the constraint-invariance analysis presented in the next section. The geometric properties required for this development are established in the following lemmas.
\begin{lemma}
Let $\boldsymbol{R}_d(t) \in \operatorname{int}(\mathcal{M})$ for all $t\in \mathbb{R}_{\geq 0}$ 
be a smooth time-varying desired attitude trajectory, and let the attitude 
error $\boldsymbol{R}_e$ be as defined in \eqref{eq:kinematics}. Then, the map 
$\varphi_t(\boldsymbol{R}_e) = \boldsymbol{R}_d(t)\boldsymbol{R}_e$ is a 
diffeomorphism of $\mathrm{SO}(3)$ for each $t \in \mathbb{R}_{\geq 0}$ satisfying 
$\varphi_t(\operatorname{int}(\mathcal{M}_e)) = \operatorname{int}(\mathcal{M})$ 
and preserves the attitude flow.
\label{lem:2}
\end{lemma}

\begin{proof}
Since $\boldsymbol{R}_d(t)\in\operatorname{int}(\mathcal{M})\subset\mathrm{SO}(3)$, the map $\varphi_t(\boldsymbol{R}_e)=\boldsymbol{R}_d(t)\boldsymbol{R}_e$ is a left translation on $\mathrm{SO}(3)$ and hence a diffeomorphism, with inverse $\varphi_t^{-1}(\boldsymbol{R})=\boldsymbol{R}_d(t)^\top\boldsymbol{R}$. Next, we show that $\varphi_t(\mathcal{M}_e)=\mathcal{M}$. For $\boldsymbol{R}=\varphi_t(\boldsymbol{R}_e)=\boldsymbol{R}_d\boldsymbol{R}_e$, we have $\boldsymbol{R}_e=\boldsymbol{R}_d^\top\boldsymbol{R}$. Thus, the constraint $\operatorname{tr}(\boldsymbol{R}_e)>-1$ is equivalent to $\operatorname{tr}(\boldsymbol{R}_d^\top\boldsymbol{R})>-1$, which defines the set $\mathcal{L}$. Furthermore, using $\boldsymbol{y}_i=\boldsymbol{R}_d^\top\boldsymbol{y}_i'$ for each $i\in\{1,...,n\}$, the pointing constraint satisfies $\boldsymbol{y}_i^\top\boldsymbol{R}_e\boldsymbol{g}_d = (\boldsymbol{R}_d^\top\boldsymbol{y}_i')^\top \boldsymbol{R}_d^\top\boldsymbol{R}\boldsymbol{g}_d  = \boldsymbol{y}_i^{\prime\top}\boldsymbol{R}\boldsymbol{g}_d \leq \cos\theta_i$, which shows that $\boldsymbol{R}\notin\mathcal{O}_i$ for all $i\in\{1,...,n\}$. Since $\boldsymbol{R}_d(t)\in\operatorname{int}(\mathcal{M})$ for all $t\in \mathbb{R}_{\geq0}$, these correspondences hold uniformly in time, yielding $\boldsymbol{R}_e\in\mathcal{M}_e$ if and only if $\varphi_t(\boldsymbol{R}_e)\in\mathcal{M}$. Hence, $\varphi_t(\mathcal{M}_e)=\mathcal{M}$. Because diffeomorphisms preserve topological interiors (see, e.g., \cite[Theorem 2.18]{Lee2012}), $\varphi_t(\operatorname{int}(\mathcal{M}_e))=\operatorname{int}(\mathcal{M})$. 

To show that the attitude flow is preserved under the diffeomorphism $\varphi_t$, consider the relation $\boldsymbol{R}=\boldsymbol{R}_d\boldsymbol{R}_e$. Differentiating both sides with respect to time and using the kinematic equations $\dot{\boldsymbol{R}}_d=\boldsymbol{R}_d\boldsymbol{\omega}^\times_d$ and \eqref{eq:kinematics}, we obtain $\dot{\boldsymbol{R}}=\boldsymbol{R}_d\left(\boldsymbol{\omega}^\times_d\boldsymbol{R}_e+\boldsymbol{R}_e\boldsymbol{e}^\times_\omega\right)=\boldsymbol{R}\boldsymbol{\omega}^\times$, where $\boldsymbol{\omega}=\boldsymbol{e}_\omega+\boldsymbol{R}_e^\top\boldsymbol{\omega}_d$. Therefore, the dynamics of the error system on $\mathcal{M}_e$ are mapped consistently to the dynamics on $\mathcal{M}$ through $\varphi_t$. As a result, every trajectory $\boldsymbol{R}_e(t)\in\operatorname{int}(\mathcal{M}_e)$ corresponds to a unique trajectory $\boldsymbol{R}(t)\in\operatorname{int}(\mathcal{M})$ for all $t\in \mathbb{R}_{\geq0}$, and vice versa through the inverse map $\varphi_t^{-1}$. This completes the proof.
\end{proof}

\begin{lemma}
Under \cref{assump:3}, the admissible sets $\mathcal{M}$ and $\mathcal{M}_e$ are smooth $3$-dimensional manifolds with boundary. Their manifold boundaries are $\partial\mathcal{M}=\bigcup_{i=1}^{n}\Sigma_i$ and $\partial\mathcal{M}_e=\bigcup_{i=1}^{n}\Sigma_i^e$, where $\Sigma_i\coloneq\{\boldsymbol{R}\in\mathrm{SO}(3):\boldsymbol{y}_i'^\top \boldsymbol{R}\boldsymbol{g}_d=\cos\theta_i\}$ and $\Sigma_i^e\coloneq\{\boldsymbol{R}_e\in\mathrm{SO}(3):\boldsymbol{y}_i^\top \boldsymbol{R}_e\boldsymbol{g}_d=\cos\theta_i\}$, $i\in \{1,\dots,n\}$. Moreover, the topological boundaries of $\mathcal{M}$ and $\mathcal{M}_e$ relative to $\mathrm{SO}(3)$ are $\partial_{\mathrm{top}}\mathcal{M}=\Gamma\cup\bigcup_{i=1}^{n}\Sigma_i$ and $\partial_{\mathrm{top}}\mathcal{M}_e=\Gamma_e\cup\bigcup_{i=1}^{n}\Sigma_i^e$, where $\Gamma\coloneq\{\boldsymbol{R}\in\mathrm{SO}(3):\operatorname{tr}(\boldsymbol{R}_d^\top\boldsymbol{R})=-1\}$ and $\Gamma_e\coloneq\{\boldsymbol{R}_e\in\mathrm{SO}(3):\operatorname{tr}(\boldsymbol{R}_e)=-1\}$.
\label{lem:3}
\end{lemma}

\begin{proof}
Since $\mathrm{SO}(3)$ is a smooth compact Lie group, it is a 3-dimensional manifold. The map $\Phi(\boldsymbol{R})=\operatorname{tr}(\boldsymbol{R}_d^\top \boldsymbol{R})$ is smooth, hence $\mathcal{L}=\Phi^{-1}(]-1,\infty[)$ is open in $\mathrm{SO}(3)$, so $\mathcal{L}$ is a smooth 3-manifold. The set $\Gamma$ is not contained in $\mathcal{M}$ but lies in its topological boundary since every neighborhood of such points intersects both $\mathcal{M}$ and its complement. For each $i$, define $h_i(\boldsymbol{R})=\boldsymbol{y}_i'^\top \boldsymbol{R}\boldsymbol{g}_d$. Then $h_i$ is smooth and $\mathcal{O}_i=h_i^{-1}(]\cos\theta,\infty[)$, hence $\mathcal{O}_i$ is open in $\mathrm{SO}(3)$. Its boundary is $\Sigma_i=h_i^{-1}(\cos\theta_i)$. For any $\boldsymbol{R}\in\Sigma_i$, the differential satisfies $Dh_i(\boldsymbol{R})[\boldsymbol{R}\boldsymbol{\omega}^\times]=\boldsymbol{y}_i'^\top \boldsymbol{R}\boldsymbol{\omega}^\times\boldsymbol{g}_d$. If $Dh_i(\boldsymbol{R})=0$, then $(\boldsymbol{R}^\top \boldsymbol{y}_i')^\top(\boldsymbol{\omega}^\times \boldsymbol{g}_d)=0$ for all $\boldsymbol{\omega}$, implying $\boldsymbol{\omega}^\top(\boldsymbol{g}_d^\times \boldsymbol{R}^\top \boldsymbol{y}_i')=0$ for all $\boldsymbol{\omega}$, hence $\boldsymbol{g}_d^\times \boldsymbol{R}^\top \boldsymbol{y}_i'=0$. Therefore $\boldsymbol{R}^\top \boldsymbol{y}_i'$ is parallel to $\boldsymbol{g}_d$, giving $\boldsymbol{y}_i'^\top \boldsymbol{R}\boldsymbol{g}_d=\pm1$, contradicting $\cos\theta_i\in(-1,1)$. Thus $Dh_i(\mathbf{R})\neq0$, therefore $\cos\theta_i$ is a regular value of $h_i$. By the regular level-set theorem (see, e.g., \cite[Corollary 5.14]{Lee2012}), $\Sigma_i$ is a smooth embedded codimension-one submanifold of $\mathrm{SO}(3)$, i.e., a smooth hypersurface.

Since the constraint regions are pairwise nonintersecting, i.e., $\Sigma_i\cap\Sigma_j=\emptyset$ for $i\neq j$ (see \cref{assump:3}), their boundaries do not intersect. Thus, each boundary point is associated with at most one constraint, ensuring that the boundary remains locally smooth and no corner singularities arise. Let $\boldsymbol{R}\in\Sigma_i$.  There exists an open neighborhood $U\subset\mathrm{SO}(3)$ of $\boldsymbol{R}$ and a smooth coordinate chart $\varphi:U\to\mathbb{R}^3$, $\varphi(\boldsymbol{R})=(x_1,x_2,x_3)$, such that $h_i(\boldsymbol{R})-\cos\theta_i=x_3$ for all $\boldsymbol{R}\in U$. Hence $\mathcal{M}\cap U=\{\boldsymbol{R}\in U:x_3\le0\}=\varphi^{-1}(\{x_3\le0\})$, showing that $\mathcal{M}$ is locally diffeomorphic to a half-space at every boundary point. Interior points inherit standard manifold charts from $\mathrm{SO}(3)$, so $\mathcal{M}$ is a smooth manifold with manifold boundary $\partial\mathcal{M}=\bigcup_{i=1}^{n}\Sigma_i$ and  topological boundary relative to  $\mathrm{SO}(3)$ is $\partial_{\mathrm{top}} \mathcal{M}=\Gamma\cup\bigcup_{i=1}^{n}\Sigma_i$. Similarly, it can be shown that $\mathcal{M}_e$ is also a smooth manifold with manifold boundary $\partial\mathcal{M}_e=\bigcup_{i=1}^{n}\Sigma_i^e$ and topological boundary relative to  $\mathrm{SO(3)}$ is $\partial_{\mathrm{top}}\mathcal{M}_e=\Gamma_e\cup\bigcup_{i=1}^{n}\Sigma_i^e$. This completes the proof.
\end{proof}

Guided by this geometric characterization, a potential function
$\Psi:\mathrm{SO}(3)\rightarrow \mathbb{R}_{\geq0}$ is constructed that promotes convergence to the desired attitude, $\boldsymbol{R}_d$, while penalizing proximity to the constraint boundaries, $\partial \mathcal{M}_e$, consistent with the geometry of $\mathrm{SO}(3)$, and is defined as
\begin{equation}
\Psi(\boldsymbol{R}_{e})\coloneq\Psi_{x}(\boldsymbol{R}_{e})-\underset{\coloneq\Psi_{y}(\boldsymbol{R}_{e})}{\underbrace{\sum_{i=1}^{n}\frac{1}{\alpha}\ln\left(\frac{\cos\theta_{i}-\boldsymbol{g}_{d}^{\top}\boldsymbol{R}_e^{\top}\boldsymbol{y}_{i}}{1+\cos(\theta_{i})}\right)}}\Psi_{x}(\boldsymbol{R}_{e}),
\label{eq:psi}
\end{equation} 
where $\Psi_{x}:\mathrm{SO(3)}\rightarrow \mathbb{R}_{\geq0}$ is defined as 
\begin{equation}
\Psi_{x}(\boldsymbol{R}_e)\coloneq2-\sqrt{1+\operatorname{tr}(\boldsymbol{R}_{e})}
\label{eq:psi_x}
\end{equation}
and it denotes the attractive potential part of $\Psi(\boldsymbol{R}_{e})$, whereas $(1-\Psi_{y}(\boldsymbol{R}_{e}))\in \mathbb{R}_{\geq0}$ represents the repulsive potential part, and $\alpha\in\mathbb{R}_{>0}$ is a design parameter. The differential-geometric properties of $\Psi(\boldsymbol{R}_e)$ are established next to facilitate the characterization of its critical points. In particular, the Riemannian gradient and Hessian provide the necessary tools for this analysis on the admissible attitude manifold $\operatorname{int}(\mathcal{M}_e)$. The following lemmas provide the necessary preliminaries for characterizing the critical points of $\Psi(\boldsymbol{R}_e)$ and analyzing their local properties.

\begin{lemma}
Consider the functions $\Psi_x(\boldsymbol{R}_e)$ and $\Psi_y(\boldsymbol{R}_e)$ as defined in (\ref{eq:psi_x}) and (\ref{eq:psi}), respectively, for all $\boldsymbol{R}_e \in \operatorname{int}(\mathcal{M}_e)$. Then, at a point $\boldsymbol{R}_e \in \operatorname{int}(\mathcal{M}_e)$,  the Riemannian gradients $\operatorname{grad}\Psi_i(\boldsymbol{R}_e) \in T_{\boldsymbol{R}_e}\mathrm{SO}(3)$, $i \in \{x, y\}$, are elements of the tangent space $T_{\boldsymbol{R}_e}\mathrm{SO}(3) = \left\{\,\boldsymbol{R}_e\boldsymbol{e}_{\omega}^{\times}: \boldsymbol{e}^{\times}_{\omega}\in\mathfrak{so}(3)\,\right\}$, and satisfy
\begin{align}
\operatorname{grad} \Psi_x(\boldsymbol{R}_e) &= -\frac{1}{4\sqrt{1 + \operatorname{tr}(\boldsymbol{R}_e)}} \left( \boldsymbol{\mathbb{I}}_{3\times3} - \boldsymbol{R}_e^{2} \right), \label{eq:grad_psi_x}\\[4pt]
\operatorname{grad} \Psi_y(\boldsymbol{R}_e) &= -\sum_{i=1}^{n}\frac{\boldsymbol{R}_e\,\operatorname{skew}\!\left(\boldsymbol{R}_e^{\top}\boldsymbol{g}_d \boldsymbol{y}_{i}^{\top}\right)}{\alpha\!\left(\cos\theta_i - \boldsymbol{g}_d^{\top}\boldsymbol{R}_e^{\top}\boldsymbol{y}_{i}\right)}. \label{eq:grad_psi_y}
\end{align}
\label{lem:4}
\end{lemma}

\begin{proof}
The differential of $\Psi_x(\boldsymbol{R}_e)$ at $\boldsymbol{R}_e$ is given by 
\[
d\Psi_x(\boldsymbol{R}_e) = -\frac{1}{2 \sqrt{1 + \operatorname{tr}(\boldsymbol{R}_e)}} \operatorname{tr}(d\boldsymbol{R}_e).
\]
By the definition of the matrix gradient, \( d\Psi_x = \operatorname{tr}((\nabla_{\boldsymbol{R}_e} \Psi_x)^\top d\boldsymbol{R}_e) \), and equating this with the earlier expression for \( d\Psi_x \) yields
\[
\operatorname{tr}((\nabla_{\boldsymbol{R}_e} \Psi_x(\boldsymbol{R}_e))^\top d\boldsymbol{R}_e) = -\frac{1}{2 \sqrt{1 + \operatorname{tr}(\boldsymbol{R}_e)}} \operatorname{tr}(d\boldsymbol{R}_e).
\]
Since this identity must hold for all variations \( d\boldsymbol{R}_e \), it follows that
\[
\nabla_{\boldsymbol{R}_e} \Psi_x(\boldsymbol{R}_e) = -\frac{1}{2 \sqrt{1 + \operatorname{tr}(\boldsymbol{R}_e)}} \boldsymbol{\mathbb{I}}_{3 \times 3}.
\]
The Riemannian gradient \(\operatorname{grad} \Psi_x(\boldsymbol{R}_e)\) is the orthogonal projection of \(\nabla_{\boldsymbol{R}_e} \Psi_x\) onto \(T_{\boldsymbol{R}_e} \mathrm{SO(3)}\) with respect to the Frobenius inner product. The projection \(\mathcal{P}_{T_{\boldsymbol{R}_e} \mathrm{SO(3)}}\) of any \(A \in \mathbb{R}^{3 \times 3}\) is given by $\mathcal{P}_{T_{\boldsymbol{R}_e} \mathrm{SO(3)}}(A) = \boldsymbol{R}_e \operatorname{skew}(\boldsymbol{R}_e^\top A)$. Applying this projection to \(A = \nabla_{\boldsymbol{R}_e} \Psi_x\), we have
\begin{equation*}
\begin{aligned}
\operatorname{grad} \Psi_x(\boldsymbol{R}_e) & = \boldsymbol{R}_e \operatorname{skew} \left( \boldsymbol{R}_e^\top \left( -\frac{1}{2 \sqrt{1 + \operatorname{tr}(\boldsymbol{R}_e)}} I_{3 \times 3} \right) \right)\\
&=-\frac{1}{4 \sqrt{1 + \operatorname{tr}(\boldsymbol{R}_e)}} (\boldsymbol{\mathbb{I}}_{3\times3} - \boldsymbol{R}_e^{2}).
\end{aligned}
\label{eq:sigma_dot_1}
\end{equation*}
Similarly, $\Psi_y(\boldsymbol{R}_e)$ satisfies the following equality:
\[
\operatorname{tr}((\nabla_{\boldsymbol{R}_e} \Psi_y(\boldsymbol{R}_e))^\top d\boldsymbol{R}_e)  = -\sum_{i=1}^{n}\frac{\operatorname{tr}\left(\boldsymbol{y}_{i} \boldsymbol{g}_d^\top d\boldsymbol{R}_e\right)}{\alpha \left(\cos\theta_i - \boldsymbol{g}_d^\top \boldsymbol{R}_e^\top \boldsymbol{y}_{i} \right)}.
\]
The Riemannian gradient of $\Psi_y (\boldsymbol{R}_e)$ is obtained as
\begin{equation*}
\begin{aligned}
\operatorname{grad} \Psi_y(\boldsymbol{R}_e)  & = \boldsymbol{R}_e \operatorname{skew}\left( \boldsymbol{R}_e^\top \nabla_{\boldsymbol{R}_e} \Psi_y(\boldsymbol{R}_e) \right)\\
 & = -\sum_{i=1}^{n} \frac{\boldsymbol{R}_e \operatorname{skew}\left( \boldsymbol{R}_e^\top \boldsymbol{g}_d \boldsymbol{y}_{i}^\top \right)}{\alpha \left(\cos\theta_i - \boldsymbol{g}_d^\top \boldsymbol{R}_e^\top \boldsymbol{y}_{i} \right)}.
\end{aligned}
\label{eq:sigma_dot_2}
\end{equation*}
Note that \(\operatorname{grad} \Psi_j(\boldsymbol{R}_e) \in T_{\boldsymbol{R}_e} \mathrm{SO}(3)\) for \(j \in \{x, y\}\), and it is well-defined for all \(\boldsymbol{R}_e \in \operatorname{int}(\mathcal{M}_e)\). This completes the proof.
\end{proof}

 \begin{lemma}
Consider the functions $\Psi_x(\boldsymbol{R}_e)$ and $\Psi_y(\boldsymbol{R}_e)$ as defined in (\ref{eq:psi_x}) and (\ref{eq:psi}), respectively, for all $\boldsymbol{R}_e \in \operatorname{int}(\mathcal{M}_e)$. Then, the following properties hold:
\begin{align}
\bigl\|\operatorname{grad}\Psi_x(\boldsymbol R_e)\bigr\| &= \sqrt{\frac{\Psi_x(\boldsymbol R_e)(\Psi_x(\boldsymbol R_e)+4)}{8}}, \label{eq:norm_grad_Psi_x} \\[4pt]
\bigl\|\operatorname{grad}\Psi_y(\boldsymbol R_e)\bigr\| &=\frac{1}{\alpha\delta}, \label{eq:norm_grad_Psi_y} 
\end{align}
where $\cos\theta_i - \boldsymbol{g}_d^\top \boldsymbol{R}_e^\top \boldsymbol{y}_{i}\geq \delta>0$, $i\in\{1,...,n\}$.
\label{lem:5}
\end{lemma}

\begin{proof}
Let $\boldsymbol R_e$ be parameterized by an axis-angle representation with rotation angle $\phi\in[0,\pi]$. Then $\operatorname{tr}(\boldsymbol R_e)=1+2\cos\phi$, and $\boldsymbol R_e^2$ corresponds to a rotation of angle $2\phi$, implying $\operatorname{tr}(\boldsymbol R_e^2)=1+2\cos2\phi$. Hence, $\|\boldsymbol{\mathbb{I}}_{3\times3}-\boldsymbol R_e^2\|^2
=2\operatorname{tr}(\boldsymbol{\mathbb{I}}_{3\times3}-\boldsymbol R_e^2)=4(1-\cos2\phi)$. Substituting the above expression into \eqref{eq:grad_psi_x} yields
$\|\operatorname{grad}\Psi_x(\boldsymbol R_e)\|=(1/\sqrt{2})\sin(\phi/2)$. Finally, using the identity $\Psi_x=2(1-\cos(\phi/2))$, we obtain the relation stated  in \eqref{eq:norm_grad_Psi_x}. To derive an expression for $\operatorname{grad}\Psi_y(\boldsymbol R_e)$ from \eqref{eq:grad_psi_y}, we use
the inequality
$\|\operatorname{skew}(\boldsymbol{R}_e^{\top}\boldsymbol{g}_d \boldsymbol{y}_{i}^{\top})\|
\le
\|\boldsymbol{R}_e^{\top}\boldsymbol{g}_d \boldsymbol{y}_{i}^{\top}\|$
together with the fact that
$\|\boldsymbol{R}_e^{\top}\boldsymbol{g}_d \boldsymbol{y}_{i}^{\top}\|
=
\|\boldsymbol{R}_e^{\top}\boldsymbol{g}_d\|\,\|\boldsymbol y_i\|
=1$,
since $\boldsymbol g_d$ and $\boldsymbol y_i$ are unit vectors and
$\boldsymbol R_e \in \mathrm{SO}(3)$. Consequently, each summand is bounded above by $(\cos\theta_i - \boldsymbol g_d^{\top}\boldsymbol R_e^{\top}\boldsymbol y_i)^{-1}$. Admissibility of $\operatorname{int}(\mathcal M_e)$ implies $\cos\theta_i-\boldsymbol g_d^{\top}\boldsymbol R_e^{\top}\boldsymbol y_i
\ge \delta>0$, from which \eqref{eq:norm_grad_Psi_y} follows. This completes the proof.
\end{proof}

\begin{lemma}
Consider the function $\Psi_x(\boldsymbol{R}_e)$ as defined in
\eqref{eq:psi_x}. Then there exist an open neighbourhood
$\mathcal{U} \coloneqq
\left\{
\boldsymbol{R}_e \in \mathcal{M}_e :
1 < \operatorname{tr}(\boldsymbol{R}_e)
\right\}$
of the identity matrix $\boldsymbol{\mathbb{I}}_{3\times3}$ and a
constant $k \in \mathbb{R}_{>0}$ such that, for every
$\boldsymbol{R}_e \in \mathcal{U}$ and every
$\boldsymbol{\xi} \in T_{\boldsymbol{R}_e}\mathrm{SO}(3)$, $\operatorname{Hess}\Psi_x(\boldsymbol{R}_e)
[\boldsymbol{\xi},\boldsymbol{\xi}]
\geq k\|\boldsymbol{\xi}\|^2$. Consequently, $\Psi_x$ is locally strongly convex on $\mathcal{U}$.
\label{lem:6}
\end{lemma}

\begin{proof}  
According to \cref{lem:4}, the Riemannian gradient of  $\Psi_y(\boldsymbol{R}_e)$ at \( \boldsymbol{R}_e \) with respect to the canonical metric on \(\mathrm{SO(3)}\), is given by \eqref{eq:grad_psi_x}. To compute the Hessian, we define $\boldsymbol{g}(\boldsymbol{R}_e) := \operatorname{grad} \Psi_x(\boldsymbol{R}_e) = c(\boldsymbol{R}_e) \, \boldsymbol{h}(\boldsymbol{R}_e)$, where
\[
c(\boldsymbol{R}_e) := \frac{1}{2 \sqrt{1 + \operatorname{tr}(\boldsymbol{R}_e)}}, \quad \boldsymbol{h}(\boldsymbol{R}_e) := \boldsymbol{R}_e \operatorname{skew}(\boldsymbol{R}_e).
\]
Differentiating $c(\boldsymbol{R}_e)$ in the direction \(\boldsymbol{\boldsymbol{\xi}} \in T_{\boldsymbol{R}_e} \mathrm{SO(3)}\), we get
\[
D c(\boldsymbol{R}_e)[\boldsymbol{\boldsymbol{\xi}}] = - \frac{1}{4 (1 + \operatorname{tr}(\boldsymbol{R}_e))^{3/2}} \operatorname{tr}(\boldsymbol{\boldsymbol{\xi}}).
\]
Taking the transpose of $\boldsymbol{\xi}$ yields $\boldsymbol{\boldsymbol{\xi}}^\top = -\boldsymbol{e}_{\omega}^{\times} \boldsymbol{R}_e^\top$. Therefore, $\operatorname{tr}(\boldsymbol{\boldsymbol{\xi}}^\top) = -\operatorname{tr}(\boldsymbol{e}_{\omega}^{\times} \boldsymbol{R}_e^\top)$. Using the cyclic property of the trace, $\operatorname{tr}(\boldsymbol{e}_{\omega}^{\times} \boldsymbol{R}_e^\top) = \operatorname{tr}(\boldsymbol{R}_e \boldsymbol{e}_{\omega}^{\times}) = \operatorname{tr}(\boldsymbol{\xi})$, which gives $\operatorname{tr}(\boldsymbol{\xi}^\top) = -\operatorname{tr}(\boldsymbol{\xi})$. Since $\operatorname{tr}(\boldsymbol{\xi}^\top) = \operatorname{tr}(\boldsymbol{\xi})$, it follows that $\operatorname{tr}(\boldsymbol{\xi}) = 0$. Thus, \(D c(\boldsymbol{R}_e)[\boldsymbol{\xi}] = 0\). The differential of $\boldsymbol{h}(\boldsymbol{R}_e)$ in direction \(\boldsymbol{\xi}\) is
\begin{align}
D \boldsymbol{h}[\boldsymbol{\xi}] 
&= \lim_{t \to 0} \frac{(\boldsymbol{R}_e + t \boldsymbol{\xi}) \operatorname{skew}(\boldsymbol{R}_e + t \boldsymbol{\xi}) - \boldsymbol{R}_e \operatorname{skew}(\boldsymbol{R}_e)}{t} \nonumber \\
&= \lim_{t \to 0} \frac{\boldsymbol{R}_e \frac{(\boldsymbol{R}_e - \boldsymbol{R}_e^\top)}{2} + t \boldsymbol{\xi} \frac{(\boldsymbol{R}_e - \boldsymbol{R}_e^\top)}{2} + t \boldsymbol{R}_e \frac{(\boldsymbol{\xi} - \boldsymbol{\xi}^\top)}{2}}{t} \nonumber \\
& \quad +\lim_{t \to 0}\frac{t^2 \boldsymbol{\xi} \frac{(\boldsymbol{\xi} - \boldsymbol{\xi}^\top)}{2} - \boldsymbol{R}_e \frac{(\boldsymbol{R}_e - \boldsymbol{R}_e^\top)}{2}}{t}\nonumber\\
&= \frac{1}{2} \left( \boldsymbol{\xi} (\boldsymbol{R}_e - \boldsymbol{R}_e^\top) + \boldsymbol{R}_e (\boldsymbol{\xi} - \boldsymbol{\xi}^\top) \right)\nonumber.
\end{align}
The Riemannian Hessian of $\Psi_x (\boldsymbol{R}_e)$ at $\boldsymbol{R}_e$ applied to the direction $\boldsymbol{\xi}$, is obtained as follows:
\begin{equation}
\begin{aligned}
\operatorname{Hess}\Psi_x(\boldsymbol{R}_e)[\boldsymbol{\xi}]  & = \nabla_{\boldsymbol{\xi}}  \boldsymbol{g}(\boldsymbol{R}_e)\\
&= \frac{\left( D \boldsymbol{g}[\boldsymbol{\xi}] -(\boldsymbol{g}(\boldsymbol{R}_e)\boldsymbol{R}_e^{\top}\boldsymbol{\xi}+\boldsymbol{\xi}\boldsymbol{R}_e^{\top}\boldsymbol{g}(\boldsymbol{R}_e))\right)}{2}\\
& = \frac{\left( \boldsymbol{R}_e (\boldsymbol{R}_e - \boldsymbol{R}_e^\top) \boldsymbol{R}_e^{\top}\boldsymbol{\xi}+  \boldsymbol{R}_e (\boldsymbol{\xi}-\boldsymbol{\xi}^{\top}) \right)}{4 \sqrt{1 + \operatorname{tr}(\boldsymbol{R}_e)}}.
\end{aligned}
\label{eq:Hess_psi_x}
\end{equation}
 It is observed that $\operatorname{Hess}\Psi_x(\boldsymbol{R}_e)[\boldsymbol{\xi}] \in T_{\boldsymbol{R}_e}\mathrm{SO(3)}$. Let $Q_1: T_{\boldsymbol{R}_e}\mathrm{SO(3)}\times T_{\boldsymbol{R}_e}\mathrm{SO(3)} \to \mathbb{R}$ be the quadratic function defined by $Q_1(\boldsymbol{\xi})\coloneqq \operatorname{Hess} \Psi_x(\boldsymbol{R}_e)[\boldsymbol{\xi},\boldsymbol{\xi}]= \langle \operatorname{Hess} \Psi_x(\boldsymbol{R}_e)[\boldsymbol{\xi}], \boldsymbol{\xi} \rangle$. Using the identity $\boldsymbol{\xi}^{\top}=-\boldsymbol{R}_e^\top\boldsymbol{\xi}\boldsymbol{R}_e^\top$, we get 
\begin{equation*}
\begin{aligned}
Q_1(\boldsymbol{\xi})  = \frac{\operatorname{tr} \left( \boldsymbol{\xi}^{\top}\boldsymbol{R}_e (\boldsymbol{R}_e - \boldsymbol{R}_e^\top) \boldsymbol{R}_e^{\top}\boldsymbol{\xi} + \boldsymbol{\xi}^\top \boldsymbol{R}_e \boldsymbol{\xi} +\boldsymbol{\xi}^\top \boldsymbol{\xi}\boldsymbol{R}_e^\top \right)}{4 \sqrt{1 + \operatorname{tr}(\boldsymbol{R}_e)}}. 
\end{aligned}
\end{equation*}
The first term can be rewritten as  $\operatorname{tr}(\boldsymbol{\xi}^{\top}\boldsymbol{R}_e (\boldsymbol{R}_e - \boldsymbol{R}_e^\top) \boldsymbol{R}_e^{\top}\boldsymbol{\xi}) = -\operatorname{tr}(\boldsymbol{e}_{\omega}^{\times} \boldsymbol{e}_{\omega}^{\times} (\boldsymbol{R}_{e} - \boldsymbol{R}_{e}^{\top}))$. Noting that $\boldsymbol{e}_{\omega}^{\times} \boldsymbol{e}_{\omega}^{\times}$ is symmetric while $(\boldsymbol{R}_{e} - \boldsymbol{R}_{e}^{\top})$ is skew-symmetric, and that the trace of the product of a symmetric and a skew-symmetric matrix is always zero, it follows that $\operatorname{tr}(\boldsymbol{\xi}^{\top}\boldsymbol{R}_e (\boldsymbol{R}_e - \boldsymbol{R}_e^\top) \boldsymbol{R}_e^{\top}\boldsymbol{\xi}) = 0$. The second term can be written as  $\operatorname{tr}(\boldsymbol{\xi}^{\top}\boldsymbol{R}_e\boldsymbol{\xi}+\boldsymbol{\xi}^{\top}\boldsymbol{\xi}\boldsymbol{R}_e^{\top}) = -\operatorname{tr}(\boldsymbol{e}_{\omega}^{\times}\boldsymbol{e}_{\omega}^{\times}(\boldsymbol{R}_e^{\top}+\boldsymbol{R}_e))$. Using the skew-symmetry property, $(\boldsymbol{e}_{\omega}^{\times})^{\top}\boldsymbol{e}_{\omega}^{\times} = \|\boldsymbol{e}_{\omega}\|^{2}\boldsymbol{\mathbb{I}}_{3\times3} - \boldsymbol{e}_{\omega}\boldsymbol{e}_{\omega}^{\top}$, we get $\operatorname{tr}(\boldsymbol{\xi}^{\top}\boldsymbol{R}_e\boldsymbol{\xi}+\boldsymbol{\xi}^{\top}\boldsymbol{\xi}\boldsymbol{R}_e^{\top}) = 2\|\boldsymbol{e}_{\omega}\|^2\operatorname{tr}(\boldsymbol{R}_e)-\boldsymbol{e}_{\omega}^{\top}(\boldsymbol{R}_e+\boldsymbol{R}_{e}^{\top})\boldsymbol{e}_{\omega}$. Using Rodrigues' formula, the rotation matrix $\boldsymbol{R}_e$ can be expressed as $\boldsymbol{R}_e = \cos\phi \, \boldsymbol{\mathbb{I}}_{3\times3}+ \sin\phi \, \boldsymbol{u}\times + (1 - \cos\phi) \, \boldsymbol{u} \boldsymbol{u}^\top$, where $\phi\in[0,\pi]$ is the rotation angle about the unit axis $\boldsymbol{u} \in \mathbb{S}^2$. The vector $\boldsymbol{e}_\omega$ can be decomposed along and perpendicular to $\boldsymbol{u}$ as $\boldsymbol{e}_\omega = (\boldsymbol{u}^\top \boldsymbol{e}_\omega)\boldsymbol{u}+ \boldsymbol{e}_\omega^\perp$, with $\boldsymbol{e}_\omega^\perp \perp \boldsymbol{u}$. Using these equations and $2\|\boldsymbol{e}_{\omega}\|^2
=\|\boldsymbol{\xi}\|^2$, we can rewrite  $Q_1(\boldsymbol{\xi})$ as follows: 
\[
Q_1=\frac{(1+\cos\phi)\|\boldsymbol{e}_\omega^\perp\|^2+2\cos\phi(\boldsymbol{u}^{\top}\boldsymbol{e}_\omega)^2)}{2 \sqrt{1 + \operatorname{tr}(\boldsymbol{R}_e)}}\geq
\frac{\cos\phi_m\|\boldsymbol \xi\|^2}{4\cos\frac{\phi_m}{2}},
\]
where $0\leq\phi\leq\phi_{m}<\pi/2$. It follows from the expression for $Q_1(\boldsymbol{\xi})$ that, for all
$\boldsymbol{R}_e \in \mathcal{U}$ and every
$\boldsymbol{\xi} \in T_{\boldsymbol{R}_e}\mathrm{SO}(3)$,
$\operatorname{Hess}\Psi_x(\boldsymbol{R}_e)[\boldsymbol{\xi},\boldsymbol{\xi}]
\geq k\|\boldsymbol{\xi}\|^2$, where
$k = \cos\phi_m/(4\cos(\phi_m/2))$. Consequently, $\Psi_x$ is locally strongly convex on $\mathcal{U}$. This completes the proof.
\end{proof}
The presence of both attractive and repulsive components in $\Psi(\boldsymbol{R}_e)$ in \eqref{eq:psi} may introduce multiple critical points within the admissible manifold $\operatorname{int}(\mathcal{M}_e)$ and does not, a priori, guarantee their nondegeneracy, potentially leading to undesired equilibria. The following theorem establishes that, under suitable conditions, $\Psi(\boldsymbol{R}_e)$ admits a unique nondegenerate minimum on the admissible manifold $\operatorname{int}(\mathcal{M}_e)$. Moreover, it will be shown that the Riemannian Hessian of $\Psi(\boldsymbol{R}_e)$ is positive definite in an open neighborhood of the desired attitude to establish local uniform strong convexity of $\Psi(\boldsymbol{R}_e)$.

\begin{theorem}
Consider the attitude error function $\Psi(\boldsymbol{R}_e)$ as defined in (\ref{eq:psi}). Then, the following properties hold:\\
(i) \( \Psi(\boldsymbol{R}_e) = 0 \) if and only if \( \boldsymbol{R}_e = \boldsymbol{\mathbb{I}}_{3\times3} \), and \( \Psi(\boldsymbol{R}_e) > 0 \) for all $\boldsymbol{R}_e\in \operatorname{int}(\mathcal{M}_e) \setminus \{\boldsymbol{\mathbb{I}}_{3\times3}\}$.\\
(ii) For every $\delta \in \mathbb{R}_{>0}$ there exists $\alpha_1^*(\delta) \in \mathbb{R}_{>0}$ such that, for any $\alpha > \alpha_1^*(\delta)$, the condition \(\operatorname{grad} \Psi(\boldsymbol{R}_e) = \boldsymbol{0}\) holds if and only if \( \boldsymbol{R}_e = \boldsymbol{\mathbb{I}}_{3\times3}\) for all $ \boldsymbol{R}_e\in \operatorname{int}(\mathcal{M}_e)$.\\
(iii) The Riemannian Hessian of $\Psi(\boldsymbol{R}_e)$ at $\boldsymbol{R}_e = \boldsymbol{\mathbb{I}}_{3\times 3}$ satisfies $\operatorname{Hess}\,\Psi(\boldsymbol{\mathbb{I}}_{3\times 3})[\boldsymbol{\xi},\boldsymbol{\xi}] > 0$ for any $\alpha\in  \mathbb{R}_{>0}$. Moreover, for every $\delta>0$ if $\alpha > \alpha_2^*(\delta)$ for some $\alpha_2^*(\delta) \in \mathbb{R}_{>0}$, then there exists an open neighborhood $\mathcal{U} \coloneqq
\left\{
\boldsymbol{R}_e \in \mathcal{M}_e :
1 < \operatorname{tr}(\boldsymbol{R}_e)
\right\}$ of the identity $\boldsymbol{\mathbb{I}}_{3\times3}$ and a constant $k \in \mathbb{R}_{>0}$ such that, for every $\boldsymbol{R}_e\in\mathcal{U}$ and every
$\boldsymbol{\xi}\in T_{\boldsymbol{R}_e}\mathrm{SO}(3)$, $\operatorname{Hess}\Psi(\boldsymbol{R}_e)
[\boldsymbol{\xi},\boldsymbol{\xi}]
\geq k\|\boldsymbol{\xi}\|^2$. Consequently, $\Psi(\boldsymbol{R}_e)$ is locally strongly convex on $\mathcal{U}$.
\label{Tem:1}
\end{theorem}

\begin{proof}
(i)  At \( \boldsymbol{R}_e = \boldsymbol{\mathbb{I}}_{3\times3} \), we have \( \mathrm{tr}(\boldsymbol{\mathbb{I}}_{3\times3}) = 3 \), so \( \Psi_x(\boldsymbol{\mathbb{I}}_{3\times3}) = 0 \), hence \( \Psi(\boldsymbol{\mathbb{I}}_{3\times3}) = 0 \).  
For \( \boldsymbol{R}_e \neq \boldsymbol{\mathbb{I}}_{3\times3} \), we have \( \mathrm{tr}(\boldsymbol{R}_e) < 3 \Rightarrow \Psi_x(\boldsymbol{R}_e) > 0 \).  
Also, since \( \boldsymbol{g}_d^\top \boldsymbol{R}_e^\top \boldsymbol{y}_{i} < \cos(\theta_i) \), the log barrier \( \Psi_y(\boldsymbol{R}_e) \) is well-defined and finite, and thus \( \Psi(\boldsymbol{R}_e) > 0 \) for all \( \boldsymbol{R}_e \in \operatorname{int}(\mathcal{M}_e) \setminus \{\boldsymbol{\mathbb{I}}_{3\times3}\} \).

(ii) For $\boldsymbol{R}_e \in \operatorname{int}(\mathcal{M}_e)$, the Riemannian gradient of $\Psi(\boldsymbol{R}_e)$ on the manifold $\mathrm{SO}(3)$ is given by
\begin{equation*}
 \operatorname{grad} \Psi(\boldsymbol{R}_e) = \operatorname{grad} \Psi_x(\boldsymbol{R}_e)(1 - \Psi_y) - \Psi_x(\boldsymbol{R}_e)\operatorname{grad} \Psi_y(\boldsymbol{R}_e).   
\end{equation*}
Suppose $\operatorname{grad} \Psi(\boldsymbol{R}_e) = \boldsymbol{0}$ for some $\boldsymbol{R}_e \in \mathcal{M}_e$. Then it must hold that $(1 - \Psi_y)\operatorname{grad} \Psi_x = \Psi_x \operatorname{grad} \Psi_y$. By the admissibility condition $\boldsymbol{R}_e \in \operatorname{int}(\mathcal{M}_e)$, we have $1 + \operatorname{tr}(\boldsymbol{R}_e) > 0$, ensuring that $\operatorname{grad} \Psi_x(\boldsymbol{R}_e)$ is well-defined and finite. Moreover, $\operatorname{grad} \Psi_x(\boldsymbol{R}_e) = \boldsymbol{0}$ if and only if $\operatorname{skew}(\boldsymbol{R}_e) = \boldsymbol{0}$, which occurs if and only if $\boldsymbol{R}_e = \boldsymbol{\mathbb{I}}_{3\times3}$. Therefore, at any $\boldsymbol{R}_e \neq \boldsymbol{\mathbb{I}}_{3\times3}$ in $\mathcal{M}_e$, $\operatorname{grad} \Psi_x(\boldsymbol{R}_e) \neq \boldsymbol{0}$. For the equation  $(1 - \Psi_y)\operatorname{grad} \Psi_x = \Psi_x \operatorname{grad} \Psi_y$ to hold at some $\boldsymbol{R}_e \neq \boldsymbol{\mathbb{I}}_{3\times3}$, the vectors $\operatorname{grad} \Psi_x$ and $\operatorname{grad} \Psi_y$ would have to be collinear. Taking the magnitude of both sides yields 
\begin{equation}
\left(\alpha - L\right) \|\operatorname{grad} \Psi_x\|= \Psi_x \|G\|,
\label{eq:grad}
\end{equation}
where $L$ and $G$ are defined as 
\begin{align*}
L &\coloneq  \sum_{i=1}^{n}\ln\left(\frac{\delta}{1+\cos\theta_{i}}\right), 
&G &\coloneq\sum_{i=1}^{n}\frac{\boldsymbol{R}_e\,\operatorname{skew}\!\left(\boldsymbol{R}_e^{\top}\boldsymbol{g}_d \boldsymbol{y}_{i}^{\top}\right)}{\!\delta},
\end{align*}
with $\cos\theta_i - \boldsymbol{g}_d^\top \boldsymbol{R}_e^\top \boldsymbol{y}_{i}\geq \delta>0$, $i\in\{1,...,n\}$. A sufficient condition for \eqref{eq:grad} to not hold is given by
\begin{equation}
\frac{\Psi_x\|G\|}{\|\operatorname{grad} \Psi_x\|}+L <\alpha. 
\end{equation}
An upper bound for the left-hand side of the equation can be obtained using \cref{lem:5} as follows: 
\begin{align*}
  \frac{\Psi_x\|G\|}{\|\operatorname{grad} \Psi_x\|}+L &\leq \max\limits_{\boldsymbol{R}_e\in\mathcal{M}_e \setminus \{\boldsymbol{\mathbb{I}}_{3\times3}\}}\left(\sqrt{\frac{8\Psi_x}{\Psi_x+4}}\|G\|+L\right)\\
  &\leq \sum_{i=1}^{n} \left( \frac{2\sqrt{6}}{3\delta} +\ln \frac{\delta}{1+\cos\theta_{i}} \right) \eqqcolon \alpha_1^*(\delta).
 \end{align*}
Consequently, for any choice of $\alpha>\alpha_1^*(\delta)$, the equality  $(1 - \Psi_y)\operatorname{grad} \Psi_x = \Psi_x \operatorname{grad} \Psi_y$ is impossible for all $\boldsymbol{R}_e \in \operatorname{int}(\mathcal{M}_e) \setminus \{\boldsymbol{\mathbb{I}}_{3\times3}\}$. Moreover, it can be observed that $\operatorname{grad}\Psi_y(\boldsymbol{R}_e)$ vanishes when $\boldsymbol{R}_e^{\top}\boldsymbol{g}_d$ is parallel to $\boldsymbol{y}_i$, since 
$\operatorname{skew}(\boldsymbol{R}_e^{\top}\boldsymbol{g}_d\,\boldsymbol{y}_i^{\top}) 
= (1/2)(\boldsymbol{R}_e^{\top}\boldsymbol{g}_d^\times \boldsymbol{y}_i)^\times = \boldsymbol{0}$, 
yet the overall gradient remains nonzero, $\operatorname{grad}\Psi(\boldsymbol{R}_e) \neq \boldsymbol{0}$, 
because $\operatorname{grad}\Psi_x(\boldsymbol{R}_e) \neq \boldsymbol{0}$ at these points. Hence, no solution exists to $\operatorname{grad} \Psi(\boldsymbol{R}_e) = 0$ for all $\boldsymbol{R}_e \in \operatorname{int}(\mathcal{M}_e) \setminus \{\boldsymbol{\mathbb{I}}_{3\times3}\}$, and the identity $\boldsymbol{R}_e = \boldsymbol{\mathbb{I}}_{3\times3}$ is the unique critical point in the admissible set, while the antipodal singular rotations $\boldsymbol{R}_e = -\boldsymbol{\mathbb{I}}_{3\times3} +2\boldsymbol{u}\boldsymbol{u}^{\top}$, with $\boldsymbol{u}\in \mathbb{S}^2$, are excluded by $\operatorname{tr}(\boldsymbol{R}_e) > -1$.

(iii) For the computation of the Riemannian Hessian of $\Psi(\boldsymbol{R}_e)$, we use Levi-Civita connection on \( \mathrm{SO(3)} \), i.e., $\mathrm{Hess} \, \Psi(\boldsymbol{R}_e)[\boldsymbol{\xi}] = \nabla_{\boldsymbol{\xi}}\mathrm{grad} \Psi(\boldsymbol{R}_e)$. Applying the Riemannian product rule, we get
\begin{align*}
\mathrm{Hess}\, \Psi(\boldsymbol{R}_e)[\boldsymbol{\xi}]&= (1 - \Psi_y(\boldsymbol{R}_e)) \cdot \mathrm{Hess}\, \Psi_x(\boldsymbol{R}_e)[\boldsymbol{\xi}] \nonumber \\
&\quad - \langle \mathrm{grad}\, \Psi_y(\boldsymbol{R}_e), \boldsymbol{\xi} \rangle \cdot \mathrm{grad}\, \Psi_x(\boldsymbol{R}_e) \nonumber \\
&\quad - \langle \mathrm{grad}\, \Psi_x(\boldsymbol{R}_e), \boldsymbol{\xi} \rangle \cdot \mathrm{grad}\, \Psi_y(\boldsymbol{R}_e) \nonumber \\
&\quad + \Psi_x(\boldsymbol{R}_e) \cdot \mathrm{Hess}\, \Psi_y(\boldsymbol{R}_e)[\boldsymbol{\xi}].
\end{align*}
Let $Q_2 : T_{\boldsymbol{R}_e}\mathrm{SO(3)} \times T_{\boldsymbol{R}_e}\mathrm{SO(3)} \to \mathbb{R}$ be the quadratic function defined by $Q_2(\boldsymbol{\xi}) \coloneqq \operatorname{Hess}\Psi(\boldsymbol{R}_e)[\boldsymbol{\xi},\boldsymbol{\xi}]=\langle \operatorname{Hess}\Psi(\boldsymbol{R}_e)[\boldsymbol{\xi}], \boldsymbol{\xi} \rangle$, which simplifies to:
\begin{align*}
Q_2(\boldsymbol{\xi})  &= \langle \operatorname{Hess}{\Psi_x(\boldsymbol{R}_e)}[\boldsymbol{\xi}], \boldsymbol{\xi} \rangle (1 - \Psi_y(\boldsymbol{R}_e)) \\
& \quad - \Psi_x(\boldsymbol{R}_e) \langle \operatorname{Hess}{\Psi_y(\boldsymbol{R}_e)}[\boldsymbol{\xi}], \boldsymbol{\xi} \rangle \\
& \quad - 2 \langle \operatorname{grad} \Psi_x(\boldsymbol{R}_e), \boldsymbol{\xi} \rangle \langle \operatorname{grad} \Psi_y(\boldsymbol{R}_e), \boldsymbol{\xi} \rangle.
\end{align*}
Since $\cos\theta_i > \boldsymbol{g}_d^\top \boldsymbol{R}_e^\top \boldsymbol{y}_i$, it follows that $(\cos\theta_i - \boldsymbol{g}_d^\top \boldsymbol{R}_e^\top \boldsymbol{y}_i)/(1 + \cos\theta_i) \in (0,1)$. Consequently, $1-\Psi_y>0$. At \(\boldsymbol{R}_e = \boldsymbol{\mathbb{I}}_{3\times3}\), we have  $\Psi_x(\boldsymbol{\mathbb{I}}_{3\times3}) = 0$, $\operatorname{grad} \Psi_x(\boldsymbol{\mathbb{I}}_{3\times3}) = 0$, and $\langle \operatorname{Hess} {\Psi_x(\boldsymbol{\mathbb{I}}_{3\times3})}[\boldsymbol{\xi}], \boldsymbol{\xi} \rangle > 0$. Hence, 
\[
Q_2(\boldsymbol{\xi})\big|_{\boldsymbol{R}_e=\boldsymbol{\mathbb{I}}_{3\times3} }= \langle \operatorname{Hess}{\Psi_x(\boldsymbol{\mathbb{I}}_{3\times3}})[\boldsymbol{\xi}], \boldsymbol{\xi} \rangle (1 - \Psi_y(\boldsymbol{\mathbb{I}}_{3\times3}))>0.
\]
Therefore, $\mathrm{Hess}\, \Psi(\boldsymbol{R}_e)[\boldsymbol{\xi},\boldsymbol{\xi}]>0$ at $\boldsymbol{R}_e=\boldsymbol{\mathbb{I}}_{3\times3}$ for all $\alpha>0$. Now, to ensure $\mathrm{Hess}\, \Psi(\boldsymbol{R}_e)[\boldsymbol{\xi},\boldsymbol{\xi}] > 0$ in an open neighborhood $\mathcal{U}\subset\mathcal{M}_e$ of the identity $\boldsymbol{\mathbb{I}}_{3\times3}$, consider that there exist \(M_1, M_2 \in \mathbb{R}_{>0} \) such that
\[
| \langle \operatorname{grad} \Psi_x, \boldsymbol{\xi} \rangle \langle \operatorname{grad} \Psi_y, \boldsymbol{\xi} \rangle | \le M_1 \|\boldsymbol{e}_\omega\|^2, 
\]
\[
\quad |\Psi_x \langle \operatorname{Hess}{\Psi_y(\boldsymbol{R}_e})[\boldsymbol{\xi}], \boldsymbol{\xi} \rangle| \le M_2 \|\boldsymbol{e}_\omega\|^2, 
\]
since \(\operatorname{grad} \Psi_x(\boldsymbol{R}_e)\), \(\operatorname{grad} \Psi_y(\boldsymbol{R}_e)\), and \(\operatorname{Hess}{\Psi_y(\boldsymbol{R}_e)}[\boldsymbol{\xi}]\) are bounded on \(\mathcal{M}\). The gradient inner products can be rewritten as $\langle \operatorname{grad} \Psi_x, \boldsymbol{\xi} \rangle = \langle \boldsymbol{v}_x, \boldsymbol{e}_\omega \rangle$ and  $\langle \operatorname{grad} \Psi_y, \boldsymbol{\xi} \rangle = \langle \boldsymbol{v}_y, \boldsymbol{e}_\omega \rangle$ for some vectors $\boldsymbol{v}_x, \boldsymbol{v}_y \in \mathbb{R}^3$. Similarly, the Hessian term admits the form $\langle \operatorname{Hess} \Psi_y[\boldsymbol{\xi}], \boldsymbol{\xi} \rangle 
= \boldsymbol{e}_\omega^\top \boldsymbol{H}_y \boldsymbol{e}_\omega$, where $\boldsymbol{H}_y \in \mathbb{R}^{3 \times 3}$ is a symmetric matrix. Applying the Cauchy--Schwarz inequality, the gradient product satisfies $|\langle \boldsymbol{v}_x, \boldsymbol{e}_\omega \rangle \langle \boldsymbol{v}_y, \boldsymbol{e}_\omega \rangle|
\leq \|\boldsymbol{v}_x\| \|\boldsymbol{v}_y\| \|\boldsymbol{e}_\omega\|^2$. Consequently, the constant $M_1$ can be taken as $M_1 = \max_{\boldsymbol{R}_e \in \operatorname{int}(\mathcal{M}_e)} \|\boldsymbol{v}_x(\boldsymbol{R}_e)\| \cdot \|\boldsymbol{v}_y(\boldsymbol{R}_e)\|$. The vector $\boldsymbol{v}_x$ for $\nabla \Psi_x$ can be obtained as 
\[
\boldsymbol{v}_x(\boldsymbol{R}_e) = \frac{(\boldsymbol{R}_e - \boldsymbol{R}_e^\top)^\vee}{4\sqrt{1 + \operatorname{tr}(\boldsymbol{R}_e)}}.
\]
Note that $\|\boldsymbol{v}_x\|=\sin^2(\phi/2)\leq1 $. Similarly, the vector $\boldsymbol{v}_y$ for $\nabla \Psi_y$ can be obtained as 
\[
\boldsymbol{v}_y = -\sum_{i=1}^{n}\frac{(1 + \cos\theta_i) }{\alpha\left( \cos\theta_i - \boldsymbol{g}_d^\top \boldsymbol{R}_e^\top \boldsymbol{y}_{i} \right)} (\boldsymbol{R}^\top_e \boldsymbol{g}_d^\times \boldsymbol{y}_{i}).
\]
Under the feasibility condition that there exists $\delta > 0$ such that $\cos\theta_i - \boldsymbol{g}_d^\top \boldsymbol{R}_e^\top \boldsymbol{y}_{i} \geq \delta$, the upper bound on $\boldsymbol{v}_y$ is given by $\|\boldsymbol{v}_y\|\leq \sum_{i=1}^{n}(1 + \cos\theta_i)/(\alpha \delta)$ as $\|(\boldsymbol{R}_e^\top\boldsymbol{g}_d)^\times \boldsymbol{y}_{i}\| \le 1$ for all $i=\{1,...,n\}$. To determine $M_2$, one has the following upper bound: 
\begin{align}
 \langle \operatorname{Hess} \Psi_y[\boldsymbol{\xi}], \boldsymbol{\xi} \rangle &\leq 
\sum_{i=1}^{n} \frac{\big\| \big( \boldsymbol{g}_d \boldsymbol{y}_{i}^\top \boldsymbol{R}_e 
- \boldsymbol{R}_e^\top \boldsymbol{y}_{i} \boldsymbol{g}_d^\top \big)^\vee \big\|^2\, \| \boldsymbol{e}_\omega \|^2}{\alpha (\cos\theta_i - \boldsymbol{g}_d^\top \boldsymbol{R}_e^\top \boldsymbol{y}_{i})^2}  \nonumber \\
&\leq \sum_{i=1}^{n}\frac{\|\boldsymbol{g}_d \times \boldsymbol{R}_e^{\top}\boldsymbol{y}_{i}\|^2 \, \| \boldsymbol{e}_\omega \|^2}
{\alpha \bigl(\cos\theta_i - \boldsymbol{g}_d^\top \boldsymbol{R}_e^\top \boldsymbol{y}_{i}\bigr)^2} \nonumber 
\end{align}
Finally, by combining the preceding results, $M_1$ and $M_2$ are obtained as  
\[
\begin{aligned}
M_1 &\leq  \sum_{i=1}^{n}\frac{(1 + \cos \theta_i)}{\alpha \delta}, 
&
M_2 &\leq   \sum_{i=1}^{n}\,\frac{2\sin^2 \theta_i}{\alpha \delta^2}.
\end{aligned}
\]
Also, there exists a constant \(c \in \mathbb{R}\) such that $\langle \operatorname{Hess}{\Psi_x(\boldsymbol{R}_e)}[\boldsymbol{\xi}], \boldsymbol{\xi} \rangle (1 - \Psi_y) \ge c \|\boldsymbol{e}_\omega\|^2$ (see \cref{lem:6}), where the expression of $c$ is obtained as follows:
\begin{align*}
c  &=\frac{\left(2\operatorname{tr}(\boldsymbol{R}_e)-\lambda_{\operatorname{max}}(\boldsymbol{R}_e+\boldsymbol{R}_e^{\top})\right)}{4 \sqrt{1 + \operatorname{tr}(\boldsymbol{R}_e)}}  \min_{\boldsymbol{R}_e \in \mathcal{M}_e} (1 - \Psi_y(\boldsymbol{R}_e))  \\
& = \frac{\cos{\phi}~\min_{\boldsymbol{R}_e \in \operatorname{int}(\mathcal{M}_e)} (1 - \Psi_y(\boldsymbol{R}_e))}{2\sqrt{2(1+\cos{\phi})}}
\end{align*}
Note that $c>0$ for all $\boldsymbol{R}_e \in \mathcal{U}$. Hence, it can be guaranteed that $Q_2(\boldsymbol{\xi})\geq (c-M_1-M_2)\|\boldsymbol{e}_\omega\|^2>0$ in the open neighborhood $\mathcal{U}\subset\mathcal{M}_e$ of $\boldsymbol{R}_e=\boldsymbol{\mathbb{I}}_{3\times3}$, if the tuning parameter $\alpha$ is chosen such that $\alpha\geq\alpha_2^{*}(\delta)$, where $\alpha_2^{*}(\delta)$ is given by 
\[
\alpha_2^{*}(\delta)\geq \sum_{i=1}^{n}\frac{4((1 + \cos \theta_i)\delta+2\sin^2 \theta_i)}{\delta^2}.
\]
Therefore, $\operatorname{Hess}\Psi(\boldsymbol{R}_e)[\boldsymbol{\xi},\boldsymbol{\xi}]
\geq k\|\boldsymbol{\xi}\|^2$,  where $k=(c-M_1-M_2)/2>0$ for all $\boldsymbol{R}_e\in \mathcal{U}$. Consequently, $\Psi(\boldsymbol{R}_e)$ is locally strongly convex on $\mathcal{U}$.
\end{proof}
\begin{remark}
The parameter $\alpha$ determines the critical-point structure of the potential function $\Psi(\boldsymbol R_e)$ on $\mathcal M_e$. By \cref{Tem:1}(ii), choosing $\alpha>\alpha_1^*(\delta)$ makes $\boldsymbol R_e=\boldsymbol{\mathbb{I}}_{3\times3}$ the unique critical point in $\operatorname{int}(\mathcal M_e)$. Moreover, \cref{Tem:1}(iii) shows that this equilibrium is the unique nondegenerate minimum for any $\alpha \in \mathbb{R}_{>0}$. If $\alpha>\alpha_2^*(\delta)$, the Hessian of $\Psi(\boldsymbol{R}_e)$ remains uniform positive definite in an open neighborhood of $\boldsymbol{\mathbb{I}}_{3\times3}$, yielding local strong convexity. Altogether, choosing $\alpha > \max\{\alpha_1^*(\delta),\alpha_2^*(\delta)\}$ ensures that $\Psi(\boldsymbol R_e)$ admits a unique nondegenerate minimum in $\operatorname{int}(\mathcal M_e)$ and is locally strongly convex about this minimum.

\end{remark}

\section{Design of Constrained  Geometric Fixed-Time  Sliding Mode Controller}
In this section, the proposed attitude potential function $\Psi(\boldsymbol{R}_e)$ (as defined in \eqref{eq:psi}) is employed to develop a constrained fixed-time geometric sliding mode control strategy for the spacecraft attitude control problem in the presence of multiple attitude forbidden zones $\mathcal{O}_{i}$, where $i \in \{1,\ldots,n\}$, and external disturbances. The proposed fixed-time geometric sliding variable $\boldsymbol{S}:\mathbb{R}^{3}\times\mathbb{R}^{3}\rightarrow\mathbb{R}^{3}$ is defined as:
\begin{equation}
\boldsymbol{S}(\boldsymbol{e}_{r},\boldsymbol{e}_{\omega})
\coloneq\boldsymbol{e}_{\omega}+\beta_{1}\boldsymbol{f}_{p}(\boldsymbol{e}_{r})+\beta_{2}\boldsymbol{e}_{r}\|\boldsymbol{e}_{r}\|^{q-1},
\label{eq:sliding_surface}
\end{equation}
where $\beta_{1},\beta_{2},p,q\in\mathbb{R}_{>0}$ satisfy  $p<1$ and $q>1$, and $\boldsymbol{e}_{r}:\mathrm{SO}(3)\times\mathrm{SO}(3)\rightarrow\mathbb{R}^3$ denotes the attitude error vector, defined as follows:
\begin{equation}
\boldsymbol{e}_{r}\coloneqq  \frac{(1-\Psi_{y}(\boldsymbol{R}_{e}))}{2\sqrt{1 + \operatorname{tr}(\boldsymbol{R}_e)}} (\boldsymbol{R}_e - \boldsymbol{R}_e^\top)^\vee  +\sum_{i=1}^{n} \frac{\Psi_{x}(\boldsymbol{R}_{e})(\boldsymbol{R}_e^{\top}\boldsymbol{g}_d)^\times \boldsymbol{y}_{i}}{\alpha(\cos\theta_{i}-\boldsymbol{g}_{d}^{\top}\boldsymbol{R}_e^{\top}\boldsymbol{y}_{i})},
\label{eq:e_r}
\end{equation}
and $\boldsymbol{f}_p(\boldsymbol{e}_r)$ is defined as
\begin{equation*}
\boldsymbol{f}_{p}
\coloneq
\begin{cases}
\displaystyle
\varepsilon^{p-1}
\left(
2-p-(1-p)
\frac{\|\boldsymbol{e}_{r}\|}{\varepsilon}
\right)
\boldsymbol{e}_{r},
&
 0\leq\|\boldsymbol{e}_{r}\|<\varepsilon;
\\[3mm]
\displaystyle
\boldsymbol{e}_{r}\|\boldsymbol{e}_{r}\|^{p-1},
&
\|\boldsymbol{e}_{r}\|\geq\varepsilon,
\end{cases}
\label{eq:regularized_p_term}
\end{equation*}
with $\varepsilon\in \mathbb{R}_{>0}$ being a small constant. 

\begin{theorem}
Under \cref{assump:2}, consider the attitude motion of the spacecraft governed by \eqref{eq:kinematics} and \eqref{eq:dynamics} together with the nonsingular geometric fixed-time sliding variable as defined in \eqref{eq:sliding_surface}, satisfying  \(\boldsymbol{S} = \mathbf{0}\). Then,  the closed-loop error state  trajectory $(\boldsymbol{R}_e(\cdot), \boldsymbol{e}_\omega(\cdot))$ converges into a sufficiently small  compact neighbourhood of $(\boldsymbol{\mathbb{I}}_{3\times 3},\mathbf{0})$ within a fixed time \(T_s \in \mathbb{R}_{\geq0}\) that satisfies
\begin{equation}
T_{s} \leq \frac{2^{(p+3)/2}}{(\alpha m)^{(p+1)/2}(1-p)} 
+ \frac{2^{(q+3)/2}}{(\alpha m)^{(q+1)/2}(q-1)},
\label{eq:T_s}
\end{equation}
where \(m\) is defined in Lemma~\ref{lem:7}. Moreover, \(\boldsymbol{R}_e = \boldsymbol{\mathbb{I}}_{3\times3}\) and \(\boldsymbol{e}_\omega = \mathbf{0}\) if and only if \(\boldsymbol{S} = \mathbf{0}\).
\label{Tem:2}
\end{theorem}
\begin{proof}
During the sliding phase, the following constraint is satisfied:
\begin{equation}
\boldsymbol{e}_{\omega}+\beta_{1}\boldsymbol{f}_{p}(\boldsymbol{e}_{r})+\beta_{2}\boldsymbol{e}_{r}\|\boldsymbol{e}_{r}\|^{q-1}=\boldsymbol{0}. 
\label{eq:sliding_phase}
\end{equation}
The attitude potential function $\Psi(\boldsymbol{R}_e)$ as  defined in 
\eqref{eq:psi}, is selected as a Lyapunov candidate function,  $V_{1}:\mathrm{SO}(3)\rightarrow\mathbb{R}_{\geq0}$, to establish 
fixed-time stability of a sufficiently small neighborhood of the desired trajectory 
$(\boldsymbol{R}_d(\cdot),\boldsymbol{\omega}_d(\cdot))$ during the sliding phase, i.e., $V_{1}(\boldsymbol{R}_e)\coloneq\Psi(\boldsymbol{R}_e)$. It is known that 
$\operatorname{tr}(\boldsymbol{P}^\top \boldsymbol{Q}\,\boldsymbol{\omega}^\times) 
= -\,\boldsymbol{\omega}^\top 
(\boldsymbol{P}^\top\boldsymbol{Q}-\boldsymbol{Q}^\top\boldsymbol{P})^\vee$ 
for any $\boldsymbol{P},\boldsymbol{Q} \in \mathbb{R}^{3\times3}$ and 
$\boldsymbol{\omega}\in\mathbb{R}^3$. Taking the time derivative of $V_1(\boldsymbol{R}_e)$ yields
\begin{align}
\dot{V}_1(\boldsymbol{R}_e) &=\langle \operatorname{grad}V_1(\boldsymbol{R}_e), \dot{\boldsymbol{R}}_e \rangle   \nonumber \\
&= \operatorname{tr} ((\operatorname{grad}V_1(\boldsymbol{R}_e))^\top \dot{\boldsymbol{R}}_e ) \nonumber \\
&= \operatorname{tr} ((\operatorname{grad}V_1(\boldsymbol{R}_e))^\top \boldsymbol{R}_e \boldsymbol{e}_\omega ^ \times ) \label{eq:1} \\
&= - \boldsymbol{e}_\omega^\top ((\operatorname{grad}V_1(\boldsymbol{R}_e))^\top \boldsymbol{R}_e -\boldsymbol{R}_e^\top \operatorname{grad}V_1(\boldsymbol{R}_e))^\vee. \nonumber
\end{align}
Now, computing the Riemannian gradient of $\Psi(\boldsymbol{R}_e)$ and applying Lemma~\ref{lem:4}, we obtain
\begin{multline*}
\operatorname{grad}V_1 =
\underbrace{\frac{(\Psi_{y}-1)\boldsymbol{R}_e\operatorname{skew}(\boldsymbol{R}_e^{\top})}
            {2\sqrt{1+\operatorname{tr}(\boldsymbol{R}_e)}}}_{\coloneqq V_{1x}(R_e)} \\
+\sum_{i=1}^{n}
 \underbrace{\frac{\Psi_x\,\boldsymbol{R}_e
 \operatorname{skew}(\boldsymbol{R}_e^\top\boldsymbol{g}_d\boldsymbol{y}_i^\top)}
 {\alpha(\cos\theta_i-\boldsymbol{g}_d^\top\boldsymbol{R}_e^\top\boldsymbol{y}_i)}}_{\coloneqq V_{1y}(R_e)}
\end{multline*}
For $V_{1x}(\boldsymbol{R}_e)$ and $V_{1y}(\boldsymbol{R}_e)$ as defined above, we can express
\begin{equation*}
    \begin{aligned}
\left\{
\begin{aligned}
V_{1x}^\top \boldsymbol{R}_e -\boldsymbol{R}_e^\top V_{1x} &= -\frac{(1-\Psi_{y}(\boldsymbol{R}_e))}{2\sqrt{1 + \operatorname{tr}(\boldsymbol{R}_e)}} (\boldsymbol{R}_e -\boldsymbol{R}_e^\top), \\
V_{1y}^\top \boldsymbol{R}_e -\boldsymbol{R}_e^\top V_{1y} &= \sum_{i=1}^{n} \frac{\Psi_{x}(\boldsymbol{y}_{i}\boldsymbol{g}_d^\top \boldsymbol{R}_e-\boldsymbol{R}_e^\top  \boldsymbol{g}_d \boldsymbol{y}_{i}^{\top})}{\alpha(\cos(\theta_{i})-\boldsymbol{g}_{d}^{\top}\boldsymbol{R}_e^{\top}\boldsymbol{y}_{i})}.
\end{aligned}
\right.
\end{aligned}
\label{eq:3}
\end{equation*}
Therefore, (\ref{eq:1}) can be rewritten as
\begin{multline*}
\dot{V}_1 = \boldsymbol{e}^\top_\omega \left(
\frac{(1-\Psi_{y})\,(\boldsymbol{R}_e - \boldsymbol{R}_e^\top)^\vee}
     {2\sqrt{1+\operatorname{tr}(\boldsymbol{R}_e)}}
\right. \\
\left.
+\; \sum_{i=1}^{n}\frac{\Psi_{x}(\boldsymbol{R}_e^\top\boldsymbol{g}_d\boldsymbol{y}_i^\top
          - \boldsymbol{y}_i\boldsymbol{g}_d^\top\boldsymbol{R}_e)^\vee}
         {\alpha(\cos\theta_i - \boldsymbol{g}_d^\top\boldsymbol{R}^\top\boldsymbol{y}_i)}
\right)
\end{multline*}
\begin{align}
&= \boldsymbol{e}^\top_\omega \left(
\frac{(1-\Psi_{y})\,(\boldsymbol{R}_e-\boldsymbol{R}_e^\top)^\vee}
     {2\sqrt{1+\operatorname{tr}(\boldsymbol{R}_e)}}
+ \sum_{i=1}^{n}
  \frac{\Psi_x(\boldsymbol{R}_e^\top\boldsymbol{g}_d)^\times\boldsymbol{y}_i}
       {\alpha(\cos\theta_i - \boldsymbol{g}_d^\top\boldsymbol{R}_e^\top\boldsymbol{y}_i)}
\right) \nonumber\\
&= \boldsymbol{e}^\top_\omega\boldsymbol{e}_r.
\label{eq:4}
\end{align}
Substituting \eqref{eq:sliding_phase} into \eqref{eq:4} and applying the inequality in \eqref{eq:psi_e_r}, together with the fact that $\min_{\boldsymbol{R}_e\in \operatorname{int}(\mathcal{M}_e)} |\Psi_y(\boldsymbol{R}_e)| = 0$, \eqref{eq:4} can be equivalently expressed as
\begin{equation}
\dot{V}_{1}
\leq
\begin{cases}
\displaystyle
-k_{1}V_{1}^{\frac{p+1}{2}}
-k_{2}V_{1}^{\frac{q+1}{2}},
&
\,V_1> A \epsilon^2;
\\[3mm]
\displaystyle
-k^{\frac{p+1}{2}}_{1}\varepsilon^{p-1}V_{1}
-k_{2}V_{1}^{\frac{q+1}{2}},
&
V_1\leq A \epsilon^2,
\end{cases}
\label{eq:V_1_dot}
\end{equation}
where $k_{1}\coloneq\beta_1(2/\alpha m)^{p+1}$ and $k_{2}\coloneq\beta_2(2/\alpha m)^{q+1}$. From \eqref{eq:V_1_dot} and Lemma~\ref{lem:1}, it follows that for $V_1 \ge A\epsilon^2$, the error state trajectory $(\boldsymbol{R}_e(\cdot), \boldsymbol{e}_\omega(\cdot))$ on the sliding manifold converges, in fixed time $T_s$ (defined in \eqref{eq:T_s}), to the sufficiently small compact neighborhood 
\begin{equation}
 \mathcal{E} \coloneq \{ (\boldsymbol{R}_e,\boldsymbol{e}_{\omega}) \in \operatorname{int}(\mathcal{M}_e)\times\mathbb{R}^3 : V_1(\boldsymbol{R}_e) \le A\epsilon^2, \|\boldsymbol{e}_{\omega}\|\leq z \} 
 \label{eq:set}
\end{equation}
of $(\boldsymbol{\mathbb{I}}_{3\times3},\textbf{0})$, where $z=\beta_1\varepsilon^{p}+\beta_2 \varepsilon^q$ and $(\boldsymbol{\mathbb{I}}_{3\times3},\textbf{0})\in \operatorname{int}(\mathcal{E})$. By Theorem~1, \(V_1(\boldsymbol{R}_e) = \Psi(\boldsymbol{R}_e) = 0\) if and only if \(\boldsymbol{R}_e = \boldsymbol{\mathbb{I}}_{3\times3}\).  For \(V_1 \leq A\epsilon^2\), \eqref{eq:V_1_dot} implies that \(\dot{V}_1(\boldsymbol{R}_e) = 0\) only when \(\boldsymbol{R}_e = \boldsymbol{\mathbb{I}}_{3\times3}\). 
Since \(\boldsymbol{e}_r = \mathbf{0}\) when \(\boldsymbol{R}_e = \boldsymbol{\mathbb{I}}_{3\times3}\), as follows from \eqref{eq:e_r}, and \eqref{eq:sliding_phase} ensures that \(\boldsymbol{e}_\omega \to \mathbf{0}\) as \(\boldsymbol{R}_e \to \boldsymbol{\mathbb{I}}_{3\times3}\). Therefore, \(\boldsymbol{R}_e = \boldsymbol{\mathbb{I}}_{3\times3}\) and \(\boldsymbol{e}_\omega = \mathbf{0}\) if and only if \(\boldsymbol{S} = \mathbf{0}\). This completes the proof.
\end{proof} 
The following lemma provides a preliminary result for the subsequent constrained attitude control law formulation and stability analysis of the closed-loop attitude control system on $\mathrm{SO}(3)\times\mathbb{R}^3$.

\begin{lemma}
Consider the attitude potential function $\Psi(\boldsymbol{R}_e)$ and the attitude error vector $\boldsymbol{e}_r$  as defined in \eqref{eq:psi} and \eqref{eq:e_r}, respectively. Assume that $\boldsymbol{g}_d^\top \boldsymbol{R}_e^\top \boldsymbol{y}_{i} \le \beta_i < \cos\theta_i$ for all $i=\{1,...,n\}$. Define the positive constants $m \coloneqq \sum_{i=1}^{n}(\cos\theta_i - \beta_i )> 0$ and $M \coloneqq \sum_{i=1}^{n}(\cos\theta_i + 1)$.  Then, the following poperties hold:\\ 
(i) There exist constants $A\coloneq\left( r_1 +2/\alpha m \right)^{-2}$ and $B\coloneq\left( r_2 -2/\alpha m \right)^{-2}$, where 
\[
\begin{aligned}
r_1 &\coloneq \max_{\boldsymbol{R}_e\in\operatorname{int}(\mathcal{M}_e)}
\left|\Psi_y(\boldsymbol{R}_e)\right|,
&
r_2 &\coloneq \min_{\boldsymbol{R}_e\in\operatorname{int}(\mathcal{M}_e)}
\left|\Psi_y(\boldsymbol{R}_e)\right|,
\end{aligned}
\]
such that 
\begin{equation}
A \|\boldsymbol{e}_r\|^2 \leq \Psi(\boldsymbol{R}_e) \leq B \|\boldsymbol{e}_r\|^2.
\label{eq:psi_e_r}
\end{equation}
(ii) The time derivative of \(\boldsymbol{e}_r\) is given by
\begin{align*}
\dot{\boldsymbol{e}}_r &= (1-\Psi_y) \left( \frac{(\boldsymbol{R}_e \boldsymbol{e}_\omega^\times + \boldsymbol{e}_\omega^\times \boldsymbol{R}_e^\top)^\vee}{2 \sqrt{1 + \operatorname{tr}(\boldsymbol{R}_e)}} - \frac{\boldsymbol{z} \boldsymbol{e}_\omega^\top\boldsymbol{z}}{4 (1 + \operatorname{tr}(\boldsymbol{R}_e))^{3/2}}  \right) \nonumber \\
&\quad  - \Psi_x \sum_{i=1}^{n}\left(\frac{\delta(\boldsymbol{R}_e^\top \boldsymbol{e}_\omega^\times \boldsymbol{g}_d)^\times \boldsymbol{y}_{i}+(\boldsymbol{R}_e^\top \boldsymbol{g}_d)^\times \boldsymbol{y}_{i} \boldsymbol{g}_d^\top \boldsymbol{R}_e^\top \boldsymbol{e}_\omega^\times \boldsymbol{y}_{i}}{\delta^2} \right) \nonumber\\
& \quad  -\sum_{i=1}^{n} \frac{\left((\boldsymbol{R}_e^\top \boldsymbol{g}_d)^\times \boldsymbol{y}_{i} \boldsymbol{e}_\omega^\top -\boldsymbol{g}_d^\top \boldsymbol{R}_e^\top \boldsymbol{e}_\omega^\times \boldsymbol{y}_{i}\right)\boldsymbol{z}}{2\alpha \sqrt{1 + \operatorname{tr}(\boldsymbol{R}_e)}\delta}, \nonumber\\
\end{align*}
where $\boldsymbol{z}\coloneq(\boldsymbol{R}_e - \boldsymbol{R}_e^\top)^\vee$ and $\delta\coloneq\sum_{i=1}^{n}(\cos\theta_i - \boldsymbol{g}_d^\top \boldsymbol{R}_e^\top \boldsymbol{y}_{i})\in [m, M]$.
\label{lem:7}
\end{lemma}

\begin{proof}
(i) The rotation matrix $\boldsymbol{R}_e$ can be parameterized by a rotation of angle $\phi \in [0,\pi]$ about an axis $\boldsymbol{u} \in \mathbb{S}^2$, so that $\operatorname{tr}(\boldsymbol{R}_e)=1+2\cos\phi$. Hence, $\Psi_x(\boldsymbol{R}_e)=4\sin^2(\phi/4)$, and the first term of $\boldsymbol{e}_r$ simplifies to
\[
\frac{\Psi_y(\boldsymbol{R}_e)}{2 \sqrt{1 + \operatorname{tr}(\boldsymbol{R}_e)}} (\boldsymbol{R}_e - \boldsymbol{R}_e^\top)^\vee = \Psi_y(\boldsymbol{R}_e) \sin \left( \frac{\phi}{2} \right)\boldsymbol{u}.
\]
Define $a:=\Psi_y(\boldsymbol{R}_e)\sin(\phi/2)$ and $b:=\Psi_x(\boldsymbol{R}_e)/(\alpha\delta)$. Since, $\|(\boldsymbol{R}_e^\top\boldsymbol{g}_d)^\times \boldsymbol{y}_{i}\| \le 1$ for all $i\in \mathcal{I}$, let  $c:=\|(\boldsymbol{R}_e^\top\boldsymbol{g}_d)^\times \boldsymbol{y}_{i}\|\in[0,1]$. Then $\boldsymbol{e}_r=a\,\boldsymbol{u}+b\,\boldsymbol{v}$, with $\|\boldsymbol{u}\|=1$ and $\|\boldsymbol{v}\|=c\le 1$, which yields
\begin{equation*}
(|a| - |b|)^2 \leq \| \boldsymbol{e}_r \|^2 \leq (|a| + |b|)^2.   \end{equation*}
Furthermore, $\Psi(\boldsymbol{R}_e)=4\sin^2(\phi/4)(1-\Psi_y(\boldsymbol{R}_e))$. Using the inequalities $\sin(\phi/4)\le \sin(\phi/2)\le 2\sin(\phi/4)$, we obtain
\[
|\Psi_y(\boldsymbol{R}_e)| \sin \left( \frac{\phi}{4} \right) \leq |a| \leq 2 |\Psi_y(\boldsymbol{R}_e)| \sin \left( \frac{\phi}{4} \right),
\]
\[
|b| \leq \frac{4 \sin^2(\phi/4)}{\alpha m}.
\]
Since $a$ and $b$ scale with $\sin(\phi/4)$ and $\sin^2(\phi/4)$, respectively, both $\|\boldsymbol{e}_r\|^2$ and $\Psi(\boldsymbol{R}_e)$ depend quadratically on $\sin(\phi/4)$. Therefore, by combining the above inequalities and using the constants $A$ and $B$ as defined in \eqref{eq:psi_e_r}, we conclude that $A\|\boldsymbol{e}_r\|^2 \le \Psi(\boldsymbol{R}_e) \le B\|\boldsymbol{e}_r\|^2$.

(ii) The attitude error vector $\boldsymbol{e}_r$ can be rewritten as \(\boldsymbol{e}_r = (1-\Psi_y(\boldsymbol{R}_e)) \boldsymbol{C} + \Psi_x(\boldsymbol{R}_e) \boldsymbol{D}\), where 
\[
\boldsymbol{C} \coloneq \left( \frac{(\boldsymbol{R}_e - \boldsymbol{R}_e^\top)^\vee}{2 \sqrt{1 + \operatorname{tr}(\boldsymbol{R}_e)}} \right),
\quad
\boldsymbol{D}\coloneq \sum_{i=1}^{n}\left( \frac{(\boldsymbol{R}_e^\top \boldsymbol{g}_d)^\times \boldsymbol{y}_{i}}{\alpha \delta} \right).
\]
Differentiating \(\boldsymbol{e}_r\) with respect to time yields
\begin{equation}
\dot{\boldsymbol{e}}_r = -\dot{\Psi}_y \boldsymbol{C} + (1-\Psi_y) \dot{\boldsymbol{C}} + \dot{\Psi}_x \boldsymbol{D} + \Psi_x \dot{\boldsymbol{D}}.
\label{eq:e_r_dot}
\end{equation}
moreover, taking the time derivative of  \(\boldsymbol{C}\) yields
\[
\dot{\boldsymbol{C}} =\frac{(\dot{\boldsymbol{R}}_e - \dot{\boldsymbol{R}}_e^\top)^\vee}{2 \sqrt{1 + \operatorname{tr}(\boldsymbol{R}_e)}} - \frac{(\boldsymbol{R}_e - \boldsymbol{R}_e^\top)^\vee \dot{\operatorname{tr}}(\boldsymbol{R}_e)}{4 (1 + \operatorname{tr}(\boldsymbol{R}_e))^{3/2}}.
\]
Using the skew-symmetric property, we obtain $(\dot{\boldsymbol{R}}_e-\dot{\boldsymbol{R}}_e^\top)^\vee=\left(\boldsymbol{R}_e\boldsymbol{e}_\omega^\times+\boldsymbol{e}_\omega^\times\boldsymbol{R}_e^\top\right)^\vee$. Consequently, the derivative of the trace is given by $\dot{\operatorname{tr}}(\boldsymbol{R}_e)=-\boldsymbol{e}_\omega^\top(\boldsymbol{R}_e-\boldsymbol{R}_e^\top)^\vee$. Taking the time derivative of \(\boldsymbol{D}\) yields
\[
\dot{\boldsymbol{D}} = \sum_{i=1}^{n} \left( \frac{(\boldsymbol{R}_e^\top \boldsymbol{g}_d)^\times \boldsymbol{y}_{i} \boldsymbol{g}_d^\top \boldsymbol{R}_e^\top \boldsymbol{e}_\omega^\times \boldsymbol{y}_{i}}{\alpha \delta^2}- \frac{(\boldsymbol{R}_e^\top \boldsymbol{e}_\omega^\times \boldsymbol{g}_d)^\times \boldsymbol{y}_{i}}{\alpha \delta}\right).
\]
Therefore, $\dot{\boldsymbol{e}}_r$ follows directly by substituting the expressions for 
$\dot{\boldsymbol{C}}$ and $\dot{\boldsymbol{D}}$ into \eqref{eq:e_r_dot}, noting that the scalar derivatives 
$\dot{\Psi}_x(\boldsymbol{R}_e)$ and $\dot{\Psi}_y(\boldsymbol{R}_e)$ are obtained via the chain rule applied to their 
respective definitions. This completes the proof.
\end{proof}

The main result of the paper is stated next. The proposed attitude control law $\boldsymbol{\tau}: (\mathrm{SO}(3))^2 \times (\mathbb{R}^3)^2 \to \mathbb{R}^3$
is defined as 
\begin{equation}
  \begin{aligned}
     \boldsymbol{\tau} &\coloneq \boldsymbol{\omega}^{\times}\boldsymbol{I}\boldsymbol{\omega} -\boldsymbol{I}\boldsymbol{e}_{\omega}^{\times}\boldsymbol{R}_{e}^{\top}\boldsymbol{\omega}_{d}-\beta_{4}\boldsymbol{S}\|\boldsymbol{S}\|^{q_{1}-1}-\beta_{3}\boldsymbol{S}\|\boldsymbol{S}\|^{p_{1}-1}\\
     & \quad +\boldsymbol{I}\boldsymbol{R}_{e}^{\top}\dot{\boldsymbol{\omega}}_{d}-\beta_{1}\boldsymbol{I}\mathcal{D}_{p}(\boldsymbol{e}_{r},
\dot{\boldsymbol{e}}_{r})-\gamma \operatorname{sgn}(\boldsymbol{S})\\
     & \quad-\beta_{2}\boldsymbol{I}\left( \mathbb{I}_{3\times3}+(q-1)\frac{\boldsymbol{e}_{r}\boldsymbol{e}_{r}^{\top}}{\|\boldsymbol{e}_{r}\|^{2}}\right) \|\boldsymbol{e}_{r}\|^{q-1}\dot{\boldsymbol{e}}_{r},
\end{aligned}
\label{eq:control_law}
\end{equation}
where $\beta_{3},\,\beta_{4} \in \mathbb{R}_{> 0}$,  $0<p_{1}< 1<q_1$, $\gamma \coloneq  \sup_{t \in \mathbb{R}_{\geq 0}} \| \boldsymbol{\tau}_d(t) \|+
([-\rho\dot\rho]_+)/\rho$, for $\lambda,\gamma_0,\varepsilon \in \mathbb{R}_{>0}$,
\begin{equation*}
\rho
\coloneq
\beta_1
\left(
\|\boldsymbol e_r\|^2+\varepsilon^2
\right)^{p/2}
+
\beta_2
\|\boldsymbol e_r\|^q
+
\frac{\lambda h}
{\sqrt{\|\boldsymbol e_r\|^2+\varepsilon^2}},
\label{eq:rho_nonsingular}
\end{equation*}
$h\coloneq\left(\Psi_{\max}-\Psi(\boldsymbol R_e)\right)$, and for $\boldsymbol{S}\coloneq[S_1~S_2~S_3]^\top$, 
\begin{equation*}
\mathcal{D}_{p}
\coloneq
\begin{cases}
\displaystyle
\varepsilon^{p-1}
\left(
2-p-(1-p)
\frac{\|\boldsymbol{e}_{r}\|}{\varepsilon}
\right)
\dot{\boldsymbol{e}}_{r}
\\[1mm]
\displaystyle\qquad
-(1-p)\varepsilon^{p-2}
\frac{
\boldsymbol{e}_{r}^{\top}\dot{\boldsymbol{e}}_{r}
}{
\|\boldsymbol{e}_{r}\|
}
\boldsymbol{e}_{r},
&
0<\|\boldsymbol{e}_{r}\|<\varepsilon,
\\[4mm]
\displaystyle
\left(
\mathbb{I}_{3\times3}
+
(p-1)
\frac{
\boldsymbol{e}_{r}\boldsymbol{e}_{r}^{\top}
}{
\|\boldsymbol{e}_{r}\|^{2}
}
\right)
\|\boldsymbol{e}_{r}\|^{p-1}
\dot{\boldsymbol{e}}_{r},
&
\|\boldsymbol{e}_{r}\|\geq\varepsilon,
\\[4mm]
\displaystyle
(2-p)\varepsilon^{p-1}
\dot{\boldsymbol{e}}_{r},
&
\boldsymbol{e}_{r}=\boldsymbol{0},
\end{cases}
\label{eq:Dp_modified}
\end{equation*}
the componentwise sign function is defined as
$\operatorname{sgn}(\boldsymbol{S})\coloneq[\operatorname{sgn}(S_1)~\operatorname{sgn}(S_2)~\operatorname{sgn}(S_3)]$ with 
\begin{equation*}
\operatorname{sgn}(S_i)
=
\begin{cases}
\phantom{-}1, & S_i>0;\\
\phantom{-}0, & S_i=0;\\
-1, & S_i<0.
\end{cases}
\end{equation*}
\begin{theorem}
Under \cref{assump:1,assump:2}, consider the spacecraft attitude motion governed by \eqref{eq:kinematics} and \eqref{eq:dynamics}, together with the nonsingular geometric fixed-time sliding variable defined in \eqref{eq:sliding_surface} and the compact set $\mathcal{E}$ defined in \eqref{eq:set}. Let the control torque $\boldsymbol{\tau}$ be given by \eqref{eq:control_law}. Then, the following properties hold:\\
(i) The closed-loop system admits a Filippov solution for all
$t \in \mathbb{R}_{\geq0}$, and the sliding variable $\boldsymbol{S}$ reaches
$\boldsymbol{0}$ within a prescribed fixed time $T_r \in \mathbb{R}_{>0}$.\\
(ii) The admissible set $\operatorname{int}(\mathcal{M})$ is the largest positively invariant subset of $\mathrm{SO}(3)$ under the closed-loop attitude dynamics.\\
(iii) The state trajectory $(\boldsymbol{R}(t), \boldsymbol{\omega}(t))$ remains in $\operatorname{int}(\mathcal{M})\times\mathbb{R}^{3}$ for all $t\in\mathbb{R}_{\geq0}$, during both the reaching and sliding phases, thereby satisfying the attitude pointing constraints at all times.\\
(iv) The closed-loop error state trajectory $(\boldsymbol{R}_e(\cdot), \boldsymbol{e}_\omega(\cdot))$ converges almost globally to the sufficiently small  compact set $\mathcal{E}$, within a prescribed fixed time. Consequently, the set $\mathcal{E}$ is almost-globally fixed-time attractive on $(\operatorname{int}(\mathcal{M}_e)\cup\mathcal{L})\times \mathbb{R}^3$.
\label{lem:3}
\end{theorem}
\begin{proof}
 (i) Substituting \eqref{eq:control_law} into \eqref{eq:dynamics} and interpreting
the discontinuous term in the Filippov sense yields the differential inclusion $\dot{\boldsymbol{x}}(t) \in F(\boldsymbol{x}(t),t)$, where $\boldsymbol{x}:\mathbb{R}_{\geq0}\rightarrow \mathrm{SO(3)}\times \mathbb{R}^3$ and
$F:\mathrm{SO(3)}\times \mathbb{R}^3\times\mathbb{R}_{\geq0}\rightrightarrows T\mathrm{SO(3)}\times \mathbb{R}^3$ is the set-valued map defined by 
\begin{equation*}
\begin{array}{ll}
\boldsymbol{x}(t)
\coloneq
\begin{bmatrix}
\boldsymbol{R}_e\\
\boldsymbol{e}_\omega
\end{bmatrix},
&
F(\boldsymbol{x},t)
\coloneq
\left\{
\begin{bmatrix}
\boldsymbol{R}_e\boldsymbol{e}_{\omega}^{\times}\\[1mm]
\boldsymbol{f}_{\omega}(\boldsymbol{x},t,\boldsymbol{v})
\end{bmatrix}
:
\boldsymbol{v}\in\operatorname{Sgn}(\boldsymbol{S})
\right\},
\end{array}
\end{equation*}
with $\operatorname{Sgn}(\boldsymbol{S})
=
\operatorname{Sgn}(S_{1})\cdot
\operatorname{Sgn}(S_{2})\cdot \operatorname{Sgn}(S_{3})$, 
\begin{equation*}
\operatorname{Sgn}(S_i)
=
\begin{cases}
\{1\}, & S_i>0;\\
[-1,1], & S_i=0;\\
\{-1\}, & S_i<0,
\end{cases}
\end{equation*}
\begin{align*}
\boldsymbol{f}_{\omega}(\boldsymbol{x},t,\boldsymbol{v})
\coloneq{}&
-\boldsymbol{I}^{-1}\beta_{4}\boldsymbol{S}\|\boldsymbol{S}\|^{q_{1}-1}
-\boldsymbol{I}^{-1}\beta_{3}\boldsymbol{S}\|\boldsymbol{S}\|^{p_{1}-1}
\nonumber\\
&-\beta_{1}\mathcal{D}_{p}(\boldsymbol{e}_{r},
\dot{\boldsymbol{e}}_{r})-\gamma\boldsymbol{I}^{-1}\boldsymbol{v}
+\boldsymbol{I}^{-1}\boldsymbol{\tau}_{d}
\nonumber\\
&-\beta_{2}
\left(
\mathbb{I}_{3\times3}
+(q-1)
\frac{\boldsymbol{e}_{r}\boldsymbol{e}_{r}^{\top}}
{\|\boldsymbol{e}_{r}\|^{2}}
\right)
\|\boldsymbol{e}_{r}\|^{q-1}\dot{\boldsymbol{e}}_{r}.
\end{align*}
The map $\operatorname{Sgn}(\cdot)$ is nonempty, compact,
convex-valued, and upper semicontinuous, whereas the remaining
state-dependent terms in $F$ are continuous on
$\operatorname{int}(\mathcal{M}_e)\times\mathbb{R}^3$. Under
Assumption~1, $\boldsymbol{\tau}_d$ is Lebesgue measurable and
essentially bounded. Hence, $F(t,\cdot)$ is nonempty, compact,
convex-valued, locally bounded, and upper semicontinuous on
$\operatorname{int}(\mathcal{M}_e)\times\mathbb{R}^3$ for all
$t\in \mathbb{R}_{\geq0}$, while $t\mapsto F(t,\boldsymbol{x})$ is measurable for each fixed $\boldsymbol{x}$. Hence, by the standard existence theorem for differential inclusions (see, e.g., \cite[Chapter~2, Section~7]{filippov1988}), a local Filippov solution exists for every initial condition $\boldsymbol{R}_e(0)\in\operatorname{int}(\mathcal{M}_e)$. Let $[0,T_{\max})$ denote the maximal interval of existence. Consider the candidate Lyapunov function $V_2(\boldsymbol{S})\coloneq\frac{1}{2}\boldsymbol{S}^\top\boldsymbol{S}$. For almost all $t\in[0,T_{\max})$ and any $\gamma>\bar{\tau}_d\geq\sup_{t \in \mathbb{R}_{\ge 0}} \| \boldsymbol{\tau}_d(t) \|$, we obtain
\begin{equation*}
	\begin{aligned}\dot{V}_{2}  
		&=  \boldsymbol{S}^{\top}\boldsymbol{I}^{-1}\left(\boldsymbol{\tau}_{d}-\beta_{3}\boldsymbol{S}\|\boldsymbol{S}\|^{p_{1}-1}-\beta_{4}\boldsymbol{S}\|\boldsymbol{S}\|^{q_{1}-1}-\gamma \operatorname{sgn}(\boldsymbol{S})\right)\\
		& \leq -\beta_{3}\|\boldsymbol{S}\|^{p_{1}+1}/l -\beta_{4}\|\boldsymbol{S}\|^{q_{1}+1}/l-\|\boldsymbol{S}\|\left(\gamma-\bar{\tau}_d\right)/l\\
		& \leq  -d_{1}V_{2}^{\frac{p_{1}+1}{2}}-d_{2}V_{2}^{\frac{q_{1}+1}{2}},
	\end{aligned}
	\label{eq:v_reachong_dot}
\end{equation*}
where $l\coloneq(\lambda_{\min}\!(\boldsymbol{I}^{\top}\boldsymbol{I}))^{0.5}$, $d_{1}\coloneq2^{\frac{p_{1}+1}{2}}\beta_{1}/l$, and
$d_{2}\coloneq2^\frac{q_{1}+1}{2}\beta_{2}/l$. Let us define $f(v)\coloneq d_1v^{(p_1+1)/2}+d_2v^{(q_1+1)/2}$, $v \in \mathbb{R}_{\geq0}$. Since $\boldsymbol{S}(\cdot)$ is absolutely continuous and $V_2\in C^1(\mathbb{R}^3)$, the composition $V_2(\boldsymbol{S}(\cdot))$ is absolutely continuous. Thus, since $\dot V_2(t)\leq-f(V_2(t))$ for almost every $t\in[0,T_{\max})$, it follows from \cite[Corollary~2.4]{matusik2020finite} that $V_2(t)\leq V_2(0)-\int_0^t f(V_2(\tau))\,d\tau$ for all $t\in[0,T_{\max})$. Consider the comparison system
$\dot\phi\coloneq-f(\phi)$, with $\phi(0)=V_2(0)$.
Then, by \cite[Proposition~2.3]{matusik2020finite}),
$V_2(t)\leq\phi(t)$ for all $t\in[0,T_{\max})$. For $\phi\in \mathbb{R}_{>0}$, the settling time of the comparison system is
$T(\phi(0))=\int_0^{\phi(0)}
[d_1\phi^{(p_1+1)/2}
+d_2\phi^{(q_1+1)/2}]^{-1}d\phi$.
Since $0<p_1<1<q_1$, this integral is bounded independently of
$\phi(0)$ by 
\[T(\phi(0))
\leq
T_r
= \frac{2}{d_1(1-p_1)}+\frac{2}{d_2(q_1-1)}.
\]
Hence, $\phi(t)=0$ for all $t\geq T_r$. Since
$0\leq V_2(t)\leq\phi(t)$, it follows that $V_2(t)=0$, and
consequently $\boldsymbol{S}(t)=\boldsymbol{0}$, for all $t\geq T_r$ on the maximal interval of existence. Therefore, the
sliding variable reaches the manifold $\boldsymbol{S}=\boldsymbol{0}$ in fixed time, with $T_r$, provided the solution exists up to $T_r$. It remains to show that $T_{\max}=+\infty$. This is established in items (ii) and (iii) below. Once global existence is established, the fixed-time estimate above holds for all $t \in \mathbb{R}_{\geq0}$.\\
(ii) According to Lemma~\ref{lem:3}, for $\boldsymbol{R}_e \in \partial_{\mathrm{top}}\mathcal{M}_e$, either $\operatorname{tr}(\boldsymbol{R}_e)=-1$ or $\boldsymbol{g}_d^\top\boldsymbol{R}_e^\top\boldsymbol{y}_i=\cos\theta_i$ for exactly one index $i\in\{1,2,\ldots,n\}$. When $\operatorname{tr}(\boldsymbol{R}_e)=-1$, we have $\operatorname{skew}(\boldsymbol{R}_e)=\boldsymbol{0}$ and $1/\sqrt{1+\operatorname{tr}(\boldsymbol{R}_e)}\to\infty$. Consequently, the geometric attitude error vector $\boldsymbol{e}_r$, and hence $\operatorname{grad}\Psi_x(\boldsymbol{R}_e)$ in \eqref{eq:grad_psi_x}, is undefined at this boundary. Moreover, at an obstacle boundary satisfying $\boldsymbol{g}_d^\top\boldsymbol{R}_e^\top\boldsymbol{y}_i=\cos\theta_i$, the repulsive potential $\Psi_y(\boldsymbol{R}_e)$ tends to $-\infty$, and consequently $\Psi(\boldsymbol{R}_e)\to+\infty$. Let $B_S
\coloneq
\frac12
\left(
\|\boldsymbol S\|^2-\rho^2
\right)$. As intially $\Psi(\boldsymbol{R}_e(0))<\Psi_{\max}<\infty$, it implies $h=\Psi_{\max}-\Psi(\boldsymbol R_e(0))>0$. Moreover, $\lambda$ can be selected such that $B_S(0)\leq0$. Substituting $\boldsymbol e_\omega$ from \eqref{eq:sliding_surface} into \eqref{eq:4} yields $\dot{\Psi}\leq r\|\boldsymbol S\|-\beta_1r^{p+1}-\beta_2r^{q+1}$, where $r=\|\boldsymbol e_r\|$. For $B_S\leq0$, we have $\|\boldsymbol S\|\leq\rho$, and hence
\begin{equation*}
\dot{\Psi}
\leq r\rho-\beta_1r^{p+1}-\beta_2r^{q+1}
=\lambda rh/(\sqrt{r^2+\varepsilon^2})\leq\lambda h.
\end{equation*}
Thus, $\dot h\geq-\lambda h$, which implies
$h(t)\geq h(0)e^{-\lambda t}>0$. Hence,
$\Psi(\boldsymbol R_e(t))<\Psi_{\max}$ whenever $B_S\leq0$. It remains to establish the forward invariance of $B_S\leq0$. Along the Filippov solutions of the $\boldsymbol S$, we get 
\begin{equation*}
\dot B_S\leq-\beta_3\|\boldsymbol S\|^{p_1+1}
-\beta_4\|\boldsymbol S\|^{q_1+1}
-(\gamma-\bar\tau_d)\|\boldsymbol S\|-\rho\dot\rho.
\label{eq:BSdot_nonsingular}
\end{equation*}
Since $\gamma-\bar\tau_d
=[-\rho\dot\rho]_+/\rho$, we obtain
\begin{align*}
\dot B_S
&\leq
-\beta_3\|\boldsymbol S\|^{p_1+1}
-\beta_4\|\boldsymbol S\|^{q_1+1}
-(\|\boldsymbol S\|/\rho)
[-\rho\dot\rho]_+
-\rho\dot\rho.
\end{align*}
On the boundary $B_S=0$, $\|\boldsymbol S\|=\rho$.  Using the fact that $[-x]_+ + x\geq 0$, it follows that $\dot B_S\leq -\beta_3\|\boldsymbol S\|^{p_1+1}-\beta_4\|\boldsymbol S\|^{q_1+1}
\leq0$ on $B_S=0$. Hence, the vector field does not point outward through the boundary, and $B_S\leq0$ is forward invariant. Combining this result with $\dot h\geq-\lambda h$ yields $\Psi(\boldsymbol R_e(t))
\leq
\Psi_{\max}-h(0)e^{-\lambda t}
<
\Psi_{\max}$ for all $t\in[0,T_{\max})$. Hence, the level set $\Omega_c \coloneq \{\boldsymbol{R}_e \in \operatorname{int}(\mathcal{M}_e) 
    : \Psi(\boldsymbol{R}_e) \leq \Psi_{\max}\}$ is forward invariant. Therefore, trajectories starting in $\Omega_c$ remain in 
$\Omega_c$ for all $t\in[0,T_{\max})$, and hence $\boldsymbol{R}_e(t) \in 
\operatorname{int}(\mathcal{M}_e)$ for all $t\in[0,T_{\max})$, establishing positive invariance of $\operatorname{int}(\mathcal{M}_e)$. To prove maximality, suppose there exists a positively invariant set 
$\mathcal{P} \supsetneq \operatorname{int}(\mathcal{M}_e)$. Then $\mathcal{P}$ must contain some point $\boldsymbol{R}_e^* \notin \operatorname{int}(\mathcal{M}_e)$. 
At any such point, either $\Psi(\boldsymbol{R}_e^*) = \infty$ or 
$\operatorname{grad}\Psi(\boldsymbol{R}_e^*)$ is undefined, so the closed-loop 
vector field is not defined at $\boldsymbol{R}_e^*$. This contradicts the assumption that $\mathcal{P}$ is a valid positively invariant set of the closed-loop system. 
Hence, no such $\mathcal{P}$ exists, and $\operatorname{int}(\mathcal{M}_e)$ is the 
largest positively invariant subset of $\mathrm{SO}(3)$ on which the closed-loop dynamics 
are well defined. By Lemma~\ref{lem:2}, the mapping $\boldsymbol{R}_e \mapsto \boldsymbol{R}$ is a diffeomorphism  that preserves the closed-loop flow and maps $\operatorname{int}(\mathcal{M}_e)$ bijectively onto $\operatorname{int}(\mathcal{M})$. Therefore, positive invariance and maximality of $\operatorname{int}(\mathcal{M}_e)$ carry over directly to 
$\operatorname{int}(\mathcal{M})$, which is the largest positively invariant subset of $\mathrm{SO}(3)$ under the closed-loop system. 

(iii)  The attitude error vector can also be expressed as $\boldsymbol{e}_r=(\operatorname{skew}(\operatorname{grad}\Psi(\boldsymbol{R}_e)))^\vee$ in the sliding variable $\boldsymbol{S}$. Consequently, $\boldsymbol{e}_\omega$ is always representable as a continuous function of $(\boldsymbol{S},\boldsymbol{e}_r)$, irrespective of whether $\boldsymbol{S}=0$ or $\boldsymbol{S}\neq 0$. Substituting this expression for $\boldsymbol{e}_\omega$ into the kinematics $\dot{\boldsymbol{R}}_e=\boldsymbol{R}_e\boldsymbol{e}_\omega^\times$, ensures that $\dot{\boldsymbol{R}}_e\in T_{\boldsymbol{R}_e}\mathrm{SO}(3)$ throughout both the reaching and sliding phases. The set $\operatorname{int}(\mathcal{M})$ has been established as the largest positively invariant subset of $\mathrm{SO}(3)$ under the closed-loop dynamics. Consequently, if $\boldsymbol{R}(0)\in\operatorname{int}(\mathcal{M})$, then $\boldsymbol{R}(t)\in\operatorname{int}(\mathcal{M})$ for all $t\in[0,T_{\max})$. Suppose now, for contradiction, that $T_{\max}<\infty$.  The rotational dynamics in~\eqref{eq:dynamics} define a smooth vector field on $\mathbb{R}^3$ for the proposed control torque $\boldsymbol{\tau}$, guaranteeing that $\boldsymbol{e}_\omega(t)$ or $\boldsymbol{\omega}(t)$ remains bounded in $\mathbb{R}^3$ for all $t\in[0,T_{\max})$.  Hence, no finite escape of the closed-loop state can occur in this maximal interval. Consequently, $(\boldsymbol{R}_e(t),\boldsymbol{e}_{\omega}(t))$ remains in a compact subset of
$\operatorname{int}(\mathcal{M})\times\mathbb{R}^{3}$ as $t\rightarrow T_{\max}^{-}$. Since $F$  is locally bounded and
upper semicontinuous on this domain, the local existence theorem can
be applied to extend the solution beyond $T_{\max}$. This contradicts
the maximality of $T_{\max}$. Hence, $T_{\max}=+\infty$. Thus, the Filippov solution exists for all $t \in \mathbb{R}_{\geq0}$. Therefore, for every initial condition $(\boldsymbol{R}(0),\boldsymbol{\omega}(0)) \in \operatorname{int}(\mathcal{M})\times\mathbb{R}^{3}$, the proposed closed-loop system admits a Filippov solution defined for all $t \in \mathbb{R}_{\geq0}$. Moreover, the corresponding trajectory remains in $\operatorname{int}(\mathcal{M})\times\mathbb{R}^{3}$ for all $t \in \mathbb{R}_{\geq0}$, thereby ensuring that the attitude pointing constraints are satisfied at all times.\\
(iv) It follows from the above that $\boldsymbol{S}(t)=\boldsymbol{0}$ for all $t\geq T_r$. Moreover, Theorem~\ref{Tem:2} implies that $(\boldsymbol{R}_e(t),\boldsymbol{e}_{\omega}(t))\in\mathcal{E}$ for all $t\geq T_s$, where $(\boldsymbol{\mathbb{I}}_{3\times3},\mathbf{0})\in\operatorname{int}(\mathcal{E})$. Therefore, the error trajectory $(\boldsymbol{R}_e(\cdot),\boldsymbol{e}_{\omega}(\cdot))$ enters the sufficiently small  compact  neighborhood $\mathcal{E}$ of $(\boldsymbol{\mathbb{I}}_{3\times3},\mathbf{0})$ within the fixed time $T=T_s+T_r$. Therefore, the set $\mathcal{E}$ is almost-globally fixed-time attractive on $(\operatorname{int}(\mathcal{M}_e)\cup\mathcal{L})\times \mathbb{R}^3$. Subsequently, the closed-loop error state trajectory converges to $(\boldsymbol{\mathbb{I}}_{3\times3},\mathbf{0})$. This completes the proof.
\end{proof}
\begin{remark}
The parameter $\delta \in \mathbb{R}_{>0}$ specifies the minimum clearance from the constraint boundary. For a prescribed safety margin $\Delta\theta_{\rm safe}>0$, selecting $\delta=\cos\theta_i-\cos(\theta_i+\Delta\theta_{\rm safe})$ ensures $\phi_i\ge\theta_i+\Delta\theta_{\rm safe}$, where $\phi_i=\cos^{-1}(\boldsymbol{g}_d^\top\boldsymbol{R}_e^\top\boldsymbol{y}_i)$ is the angular separation between $\boldsymbol{R}_e\boldsymbol{g}_d$ and the forbidden direction $\boldsymbol{y}_i$. Increasing $\delta$ improves barrier conditioning and safety clearance, whereas decreasing $\delta$ permits closer operation at the cost of larger barrier gradients and control gains.
\end{remark}

\section{Results and Discussions}

\begin{figure*}
	
	\begin{center}
		
		\includegraphics[width=0.6\columnwidth,totalheight=0.20\paperwidth]{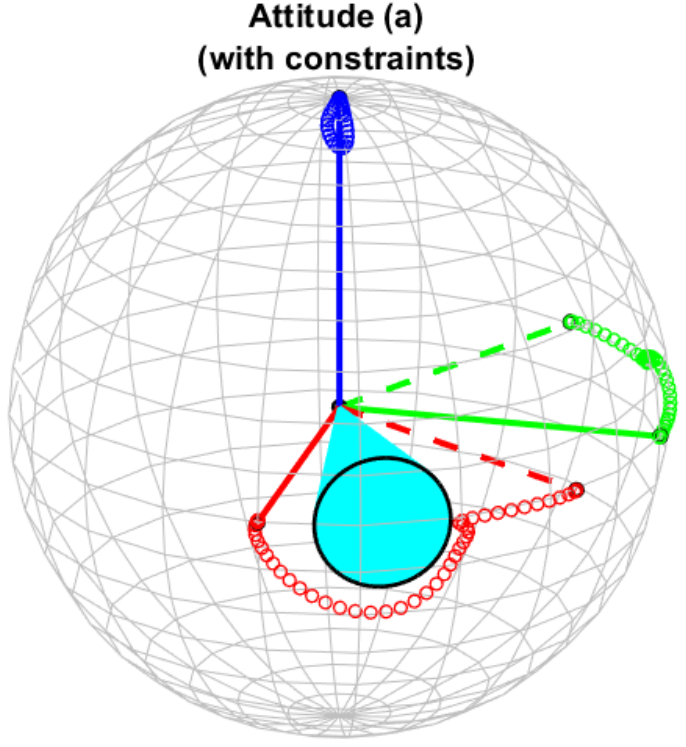}\includegraphics[width=0.7\columnwidth,totalheight=0.18\paperwidth]{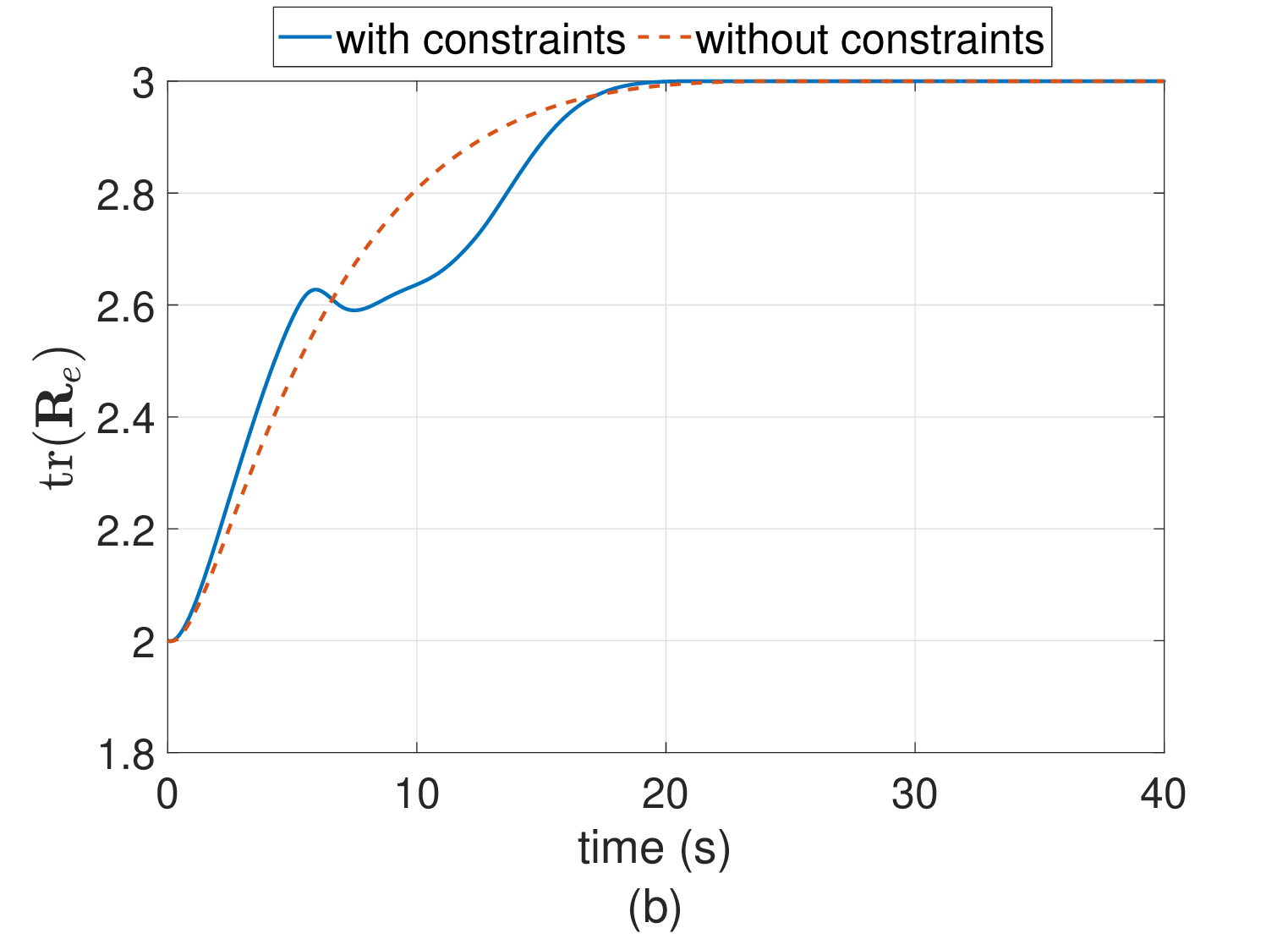}\includegraphics[width=0.6\columnwidth,totalheight=0.20\paperwidth]{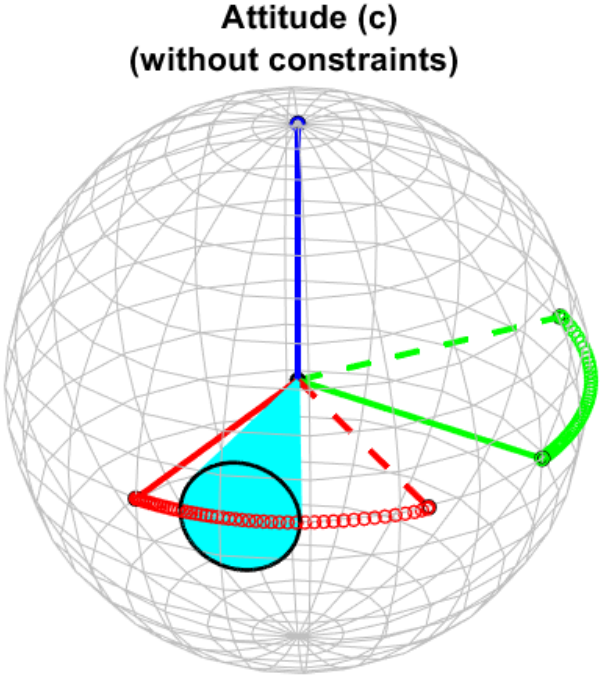}
        \includegraphics[width=0.7\columnwidth,totalheight=0.20\paperwidth]{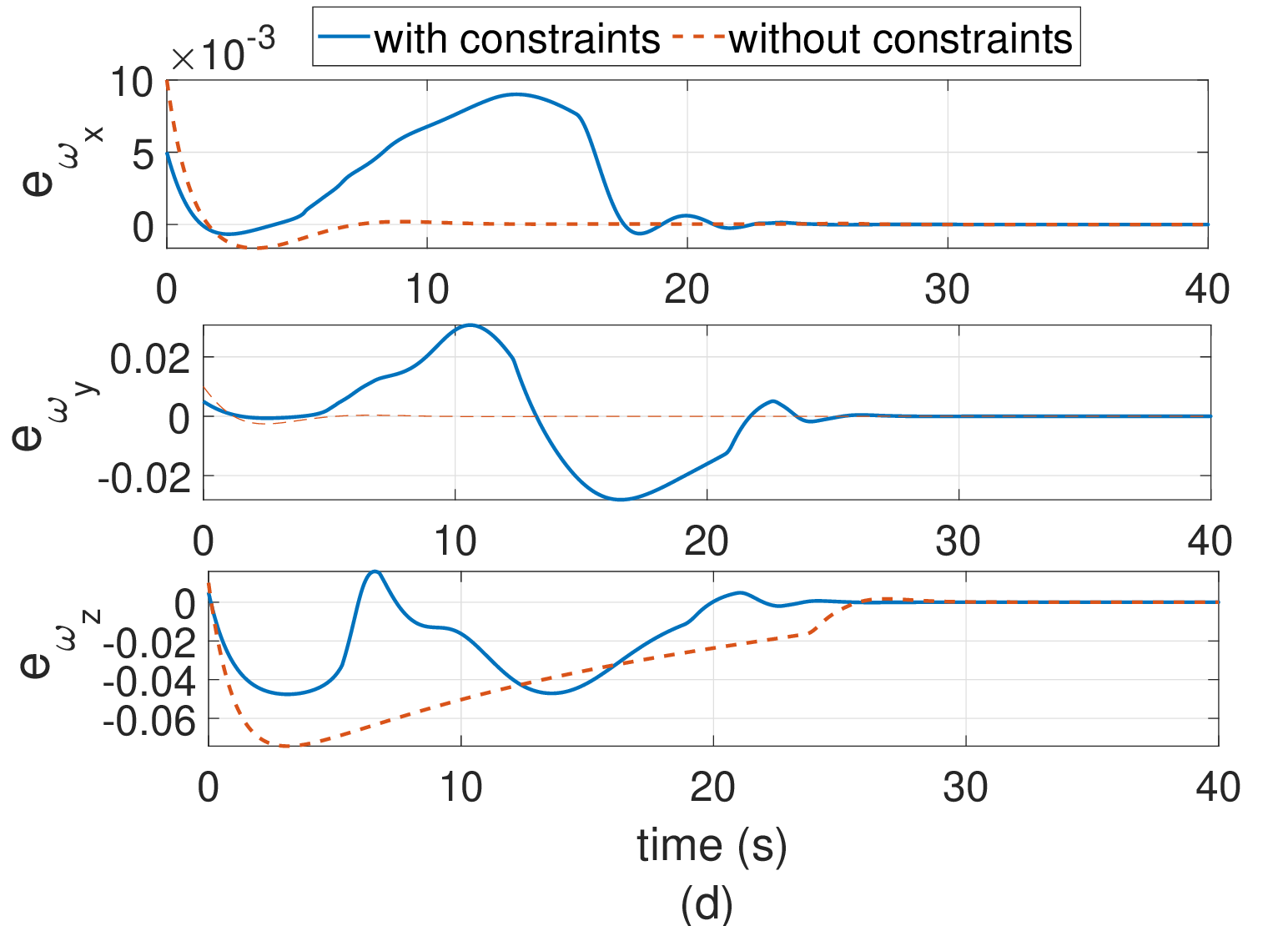}\includegraphics[width=0.7\columnwidth,totalheight=0.20\paperwidth]{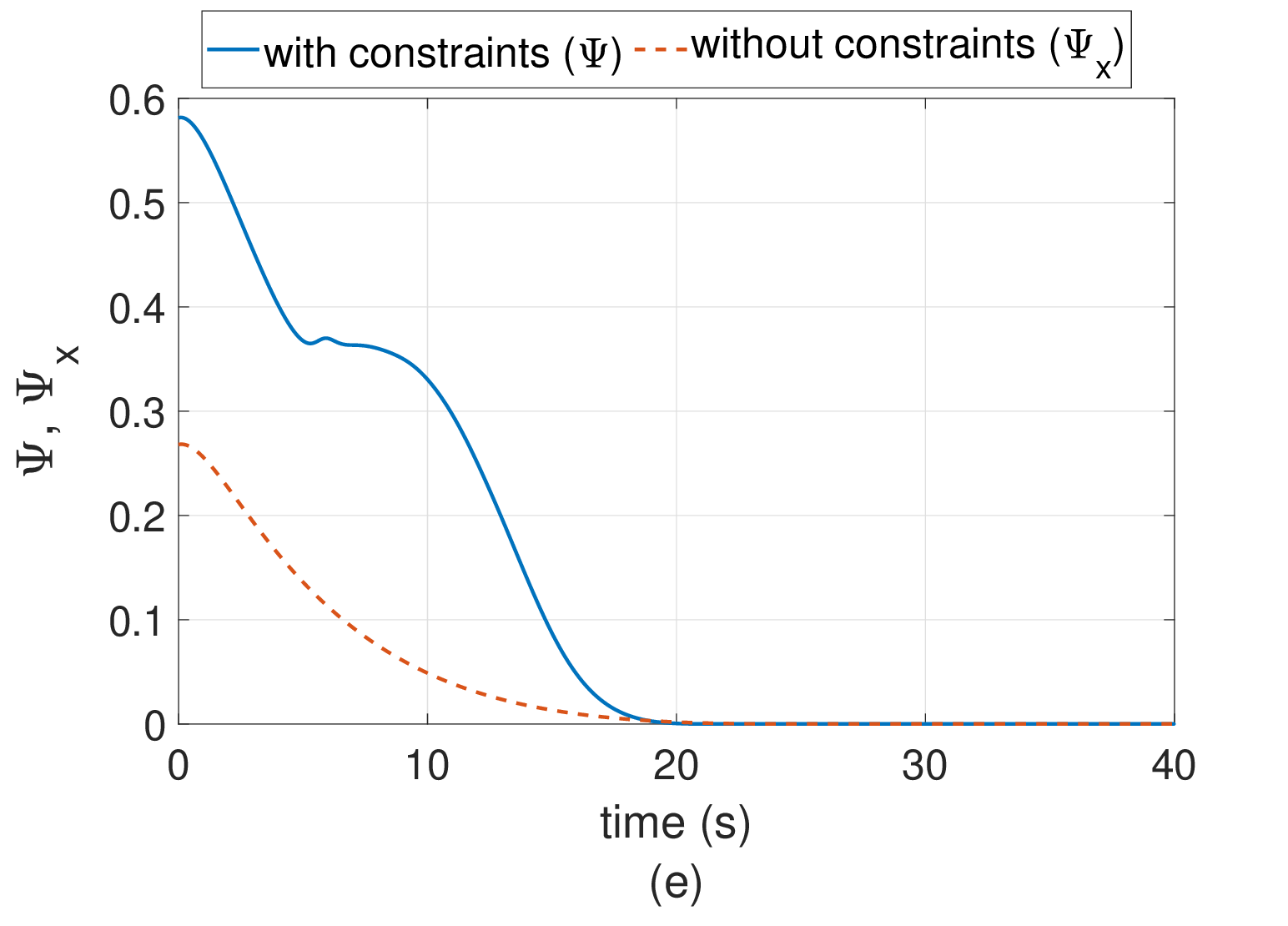}\includegraphics[width=0.7\columnwidth,totalheight=0.20\paperwidth]{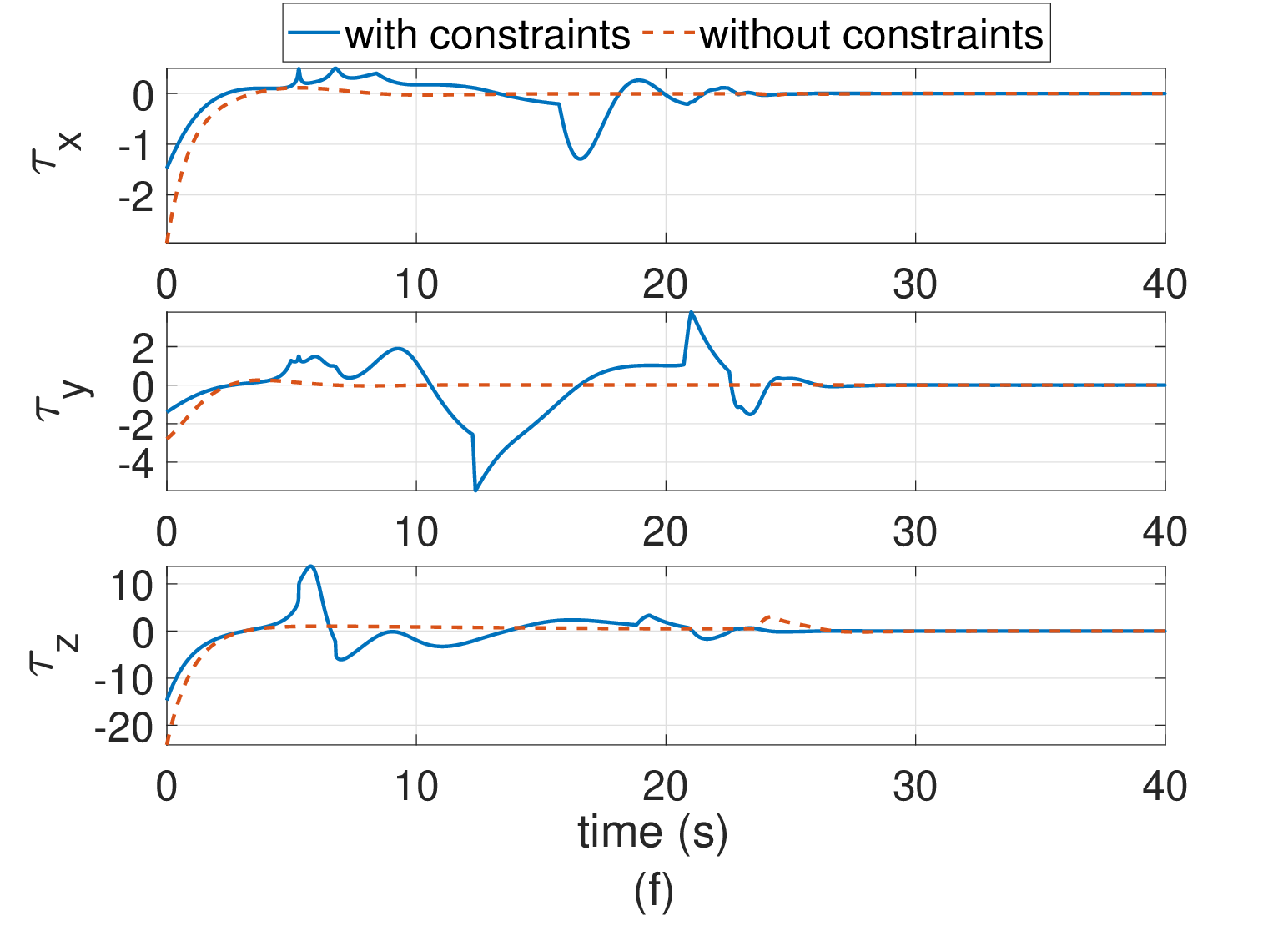}
        \includegraphics[width=\textwidth]{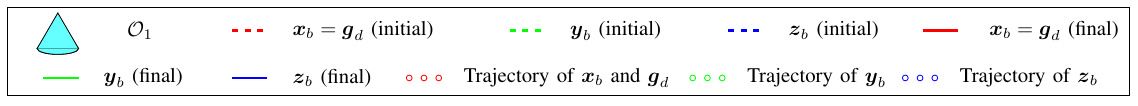}
	\end{center}
	\caption{Simulation results for Example~1: 
(a) and (c) Attitude orientation trajectories of the body-fixed frame $\mathcal{F}_{\mathcal{B}}=\{\boldsymbol{x}_b,\boldsymbol{y}_b,\boldsymbol{z}_b\}$ and the boresight vector $\boldsymbol{g}_d$, 
(b) Time response of $\mathrm{tr}(\boldsymbol{R}_e)$, 
(d) Time response of the angular velocity error vector $\boldsymbol{e}_{\omega}\coloneq[\,e_{\omega_x}\; e_{\omega_y}\; e_{\omega_z}\,]^{\top}$ ($\mathrm{rad/s}$),
(e) Time responses of $\Psi$ and $\Psi_x$, 
(f) Time response of the control input vector $\boldsymbol{\tau}\coloneq[\,\tau_x\; \tau_y\; \tau_z\,]^{\top}$ ($\mathrm{N \cdot m}$).}
\label{fig:1}
\end{figure*}

Numerical simulations are presented in this section to evaluate the efficacy of the proposed fixed-time geometric sliding mode control approach for spacecraft attitude control subject to multiple attitude pointing constraints. The inertia matrix of the spacecraft is chosen as $\boldsymbol{I}=\textrm{diag}(250,\,400,\,300)\,$ ($\mathrm{kg\cdot m^{2}}$). The external disturbance torque is considered as $\boldsymbol{\tau}_d(t)=[\begin{array}{@{}ccc@{}}
0.5\sin(\pi t/5) & 0.5\cos(\pi t/10) & 0.5\sin(\pi t/20)\end{array}]^{\top}$ ($\mathrm{N \cdot m}$). Without loss of generality, we assume that the sensitive instrument is mounted in such a way that its normalized boresight vector is aligned with the body-fixed $\boldsymbol{x}_b$ axis of the frame $\mathcal{F}_{\mathcal{B}}=\{\boldsymbol{x}_b,\boldsymbol{y}_b,\boldsymbol{z}_b\}$, i.e., $\boldsymbol{g}_d = \boldsymbol{x}_b = [\begin{array}{@{}ccc@{}}
1 & 0 & 0\end{array}]^{\top}$. The effectiveness of the proposed control law is evaluated through the following two numerical examples:

\noindent\textbf{Example 1:}  In this example, only one forbidden zone is considered, with the unit vector along the forbidden direction being chosen as  $\boldsymbol{y}'_1=[\begin{array}{@{}ccc@{}}
0.9276 & 0.3736 & 0\end{array}]^{\top}$ in $\mathcal{F}_{\mathcal{N}}$ frame. The minimum allowable angle between $\boldsymbol{R}\boldsymbol{g}_d$ and $\boldsymbol{y}'_1$ is chosen as $\theta_1=12^\circ$. The initial attitude and angular velocity of the spacecraft are chosen as $\boldsymbol{R}(0)=\operatorname{exp}\left((\pi/3)\boldsymbol{e}^\times_3\right)$ and $\boldsymbol{\omega}(0)=[\begin{array}{@{}ccc@{}}
0.01 & 0.01 & 0.01\end{array}]^{\top}$ ($\mathrm{rad/s}$), respectively, where $\boldsymbol{e}_3=[\begin{array}{@{}ccc@{}}
0 & 0 & 1\end{array}]^{\top}$. The desired attitude and angular velocity are chosen as $\boldsymbol{R}_d=\boldsymbol{\mathbb{I}}_{3\times3}$ and $\boldsymbol{\omega}_d=[\begin{array}{@{}ccc@{}}
0 & 0 & 0\end{array}]^{\top}$, respectively. Figure \ref{fig:1} illustrates the spacecraft attitude maneuver trajectories under both the constraint-free and constrained simulations. Figure \ref{fig:1}(c) illustrates the attitude evolution trajectorries in the absence of constraint, while Fig. \ref{fig:1}(a) shows the corresponding motion when geometric attitude constraint is enforced. The shaded cone represents the forbidden pointing zone, within which the body-fixed sensor axis, i.e., $\boldsymbol{g}_d = \boldsymbol{x}_b=[\begin{array}{@{}ccc@{}}
1 & 0 & 0\end{array}]^{\top}$ is not allowed to enter. In the unconstrained case, the attitude trajectory follows the shortest rotation path to the desired attitude. As a result, the path intersects to the interior of the exclusion cones, since no mechanism prevents the spacecraft from pointing temporarily through these regions. In contrast, the constrained trajectory remains entirely outside the forbidden cones. As the motion approaches a constraint boundary, the controller alters the direction of rotation and steers the attitude around the restricted region before continuing towards the target. This produces a visibly curved and elongated trajectory on the manifold compared with the unconstrained case. Nevertheless, both trajectories eventually reach the same final desired attitude. The plots, therefore, highlight the main effect of the attitude pointing constraint, namely that it reshapes the feasible attitude path by avoiding the exclusion zones while still preserving the final tracking objective.

Figure~\ref{fig:1}(b) illustrates the time history of $\mathrm{tr}(\boldsymbol{R}_e)$. Since $\mathrm{tr}(\boldsymbol{R}_e)=3$
corresponds to perfect attitude alignment, both cases ultimately converge to the same terminal value, confirming fixed-time attitude tracking. However, the transient responses differ significantly. In the unconstrained case, the trace increases gradually, following the shortest rotation path to the desired
attitude. When attitude constraints are imposed, the trace rises much more rapidly in the initial phase, but then briefly saturates as the controller modifies the rotation to avoid the forbidden region before completing convergence. This behavior indicates that the constrained controller does not always allow direct geodesic motion on $\mathrm{SO}(3)$; but instead, it enforces a detour that respects the exclusion geometry. Figure~\ref{fig:1}(e) presents the corresponding configuration error functions. The unconstrained case uses the standard attitude error potential $\Psi_x$, which decreases monotonically as the system follows the shortest feasible path to the equilibrium. In contrast, the constrained case employs a modified
potential function $\Psi$ that incorporates the geometric forbidden zone. Initially, $\Psi$ takes a larger value, and its decay is steeper as the controller rapidly drives the state away from the boundary of the exclusion cone. Near the constraint boundary, the rate of decrease temporarily slows, indicating that the controller prioritizes constraint satisfaction before
continuing the descent toward the minimum. Ultimately, both error functions converge to zero, showing that constraint enforcement mainly affects the transient evolution of the configuration error while preserving the same desired attitude and guaranteeing stability of the closed-loop attitude control system.

The time histories of the angular velocity error vector ($\boldsymbol{e}_\omega$) and the control torque vector ($\boldsymbol{\tau}$) for both the constrained and unconstrained cases are shown in Figs.  \ref{fig:1}(d) and  \ref{fig:1}(f), respectively. In the absence of attitude constraints, the angular-velocity errors converge smoothly to zero along a direct rotation path, resulting in small overshoot and modest control effort. In contrast, when attitude constraints are enforced, larger transient deviations are observed. These excursions arise because the controller modifies the rotation trajectory to remain within the admissible attitude region before completing the maneuver. Consequently, the settling time increases slightly and the angular-velocity profiles exhibit secondary oscillations associated with constraint-avoidance maneuvers. A similar behavior is observed in the control torques. The unconstrained case generates smooth commands with relatively small magnitudes, whereas the constrained case requires larger peak torques and several corrective actions when the trajectory is steered away from restricted orientations. Despite these transient differences, both strategies ultimately drive the tracking error to zero, demonstrating that constraint enforcement primarily affects the transient dynamics and control effort rather than the steady-state performance.

\begin{figure*}
	
	\begin{center}
		
		\includegraphics[width=0.6\columnwidth,totalheight=0.20\paperwidth]{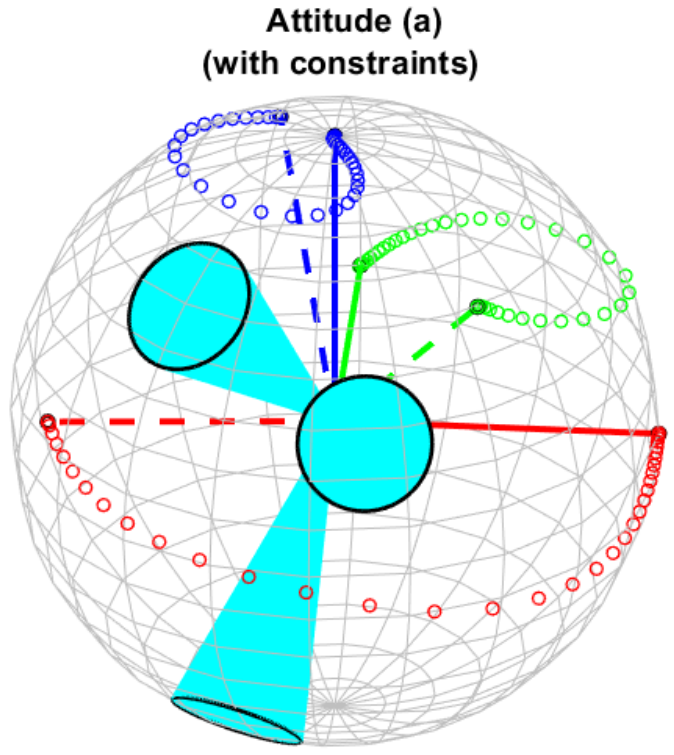}\includegraphics[width=0.7\columnwidth,totalheight=0.18\paperwidth]{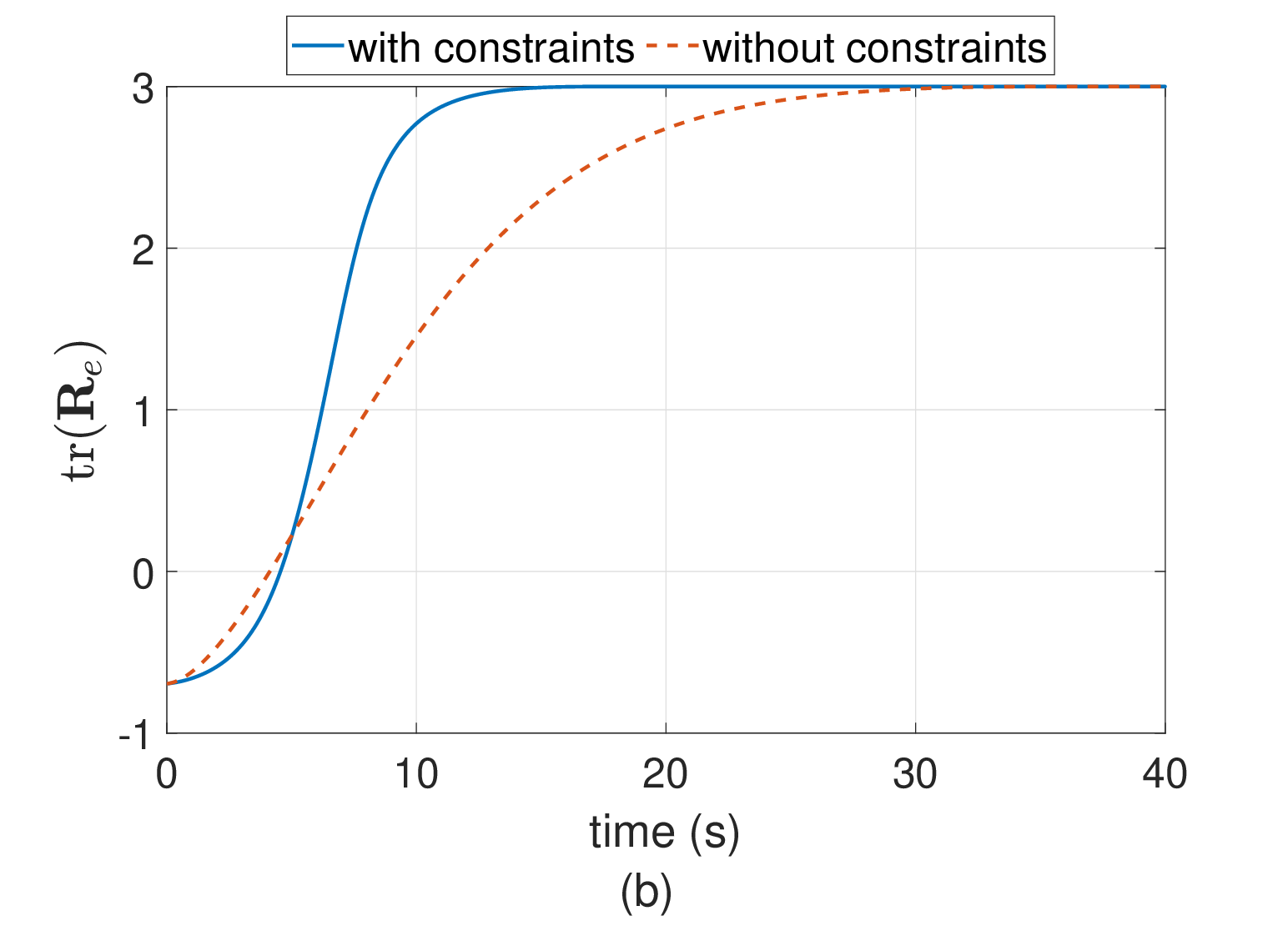}\includegraphics[width=0.6\columnwidth,totalheight=0.20\paperwidth]{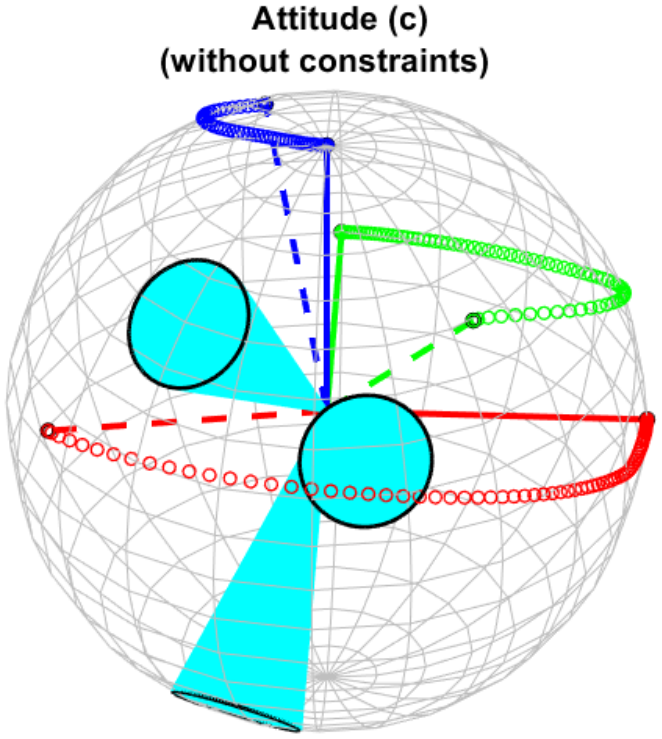}
        \includegraphics[width=0.7\columnwidth,totalheight=0.20\paperwidth]{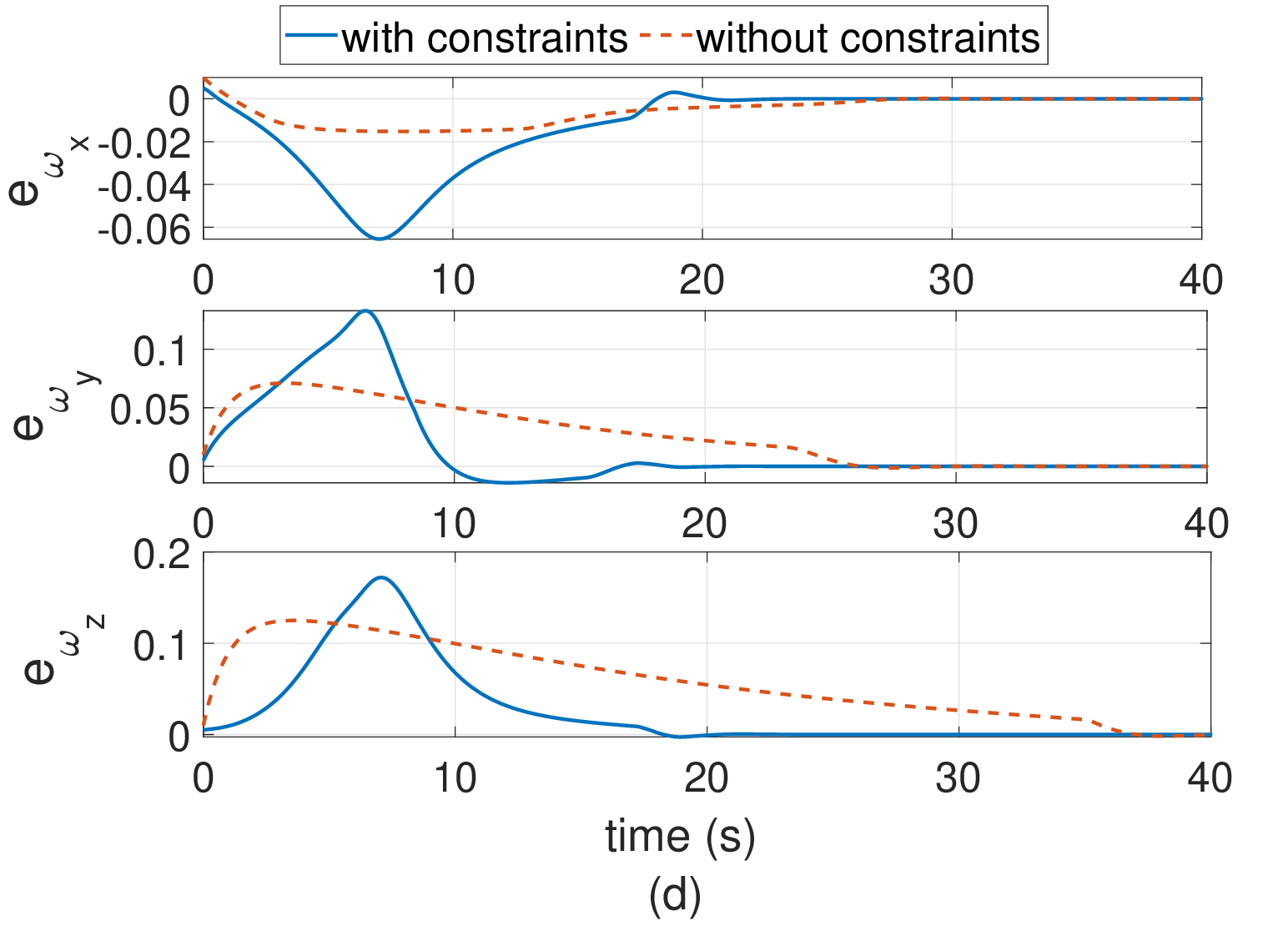}\includegraphics[width=0.7\columnwidth,totalheight=0.20\paperwidth]{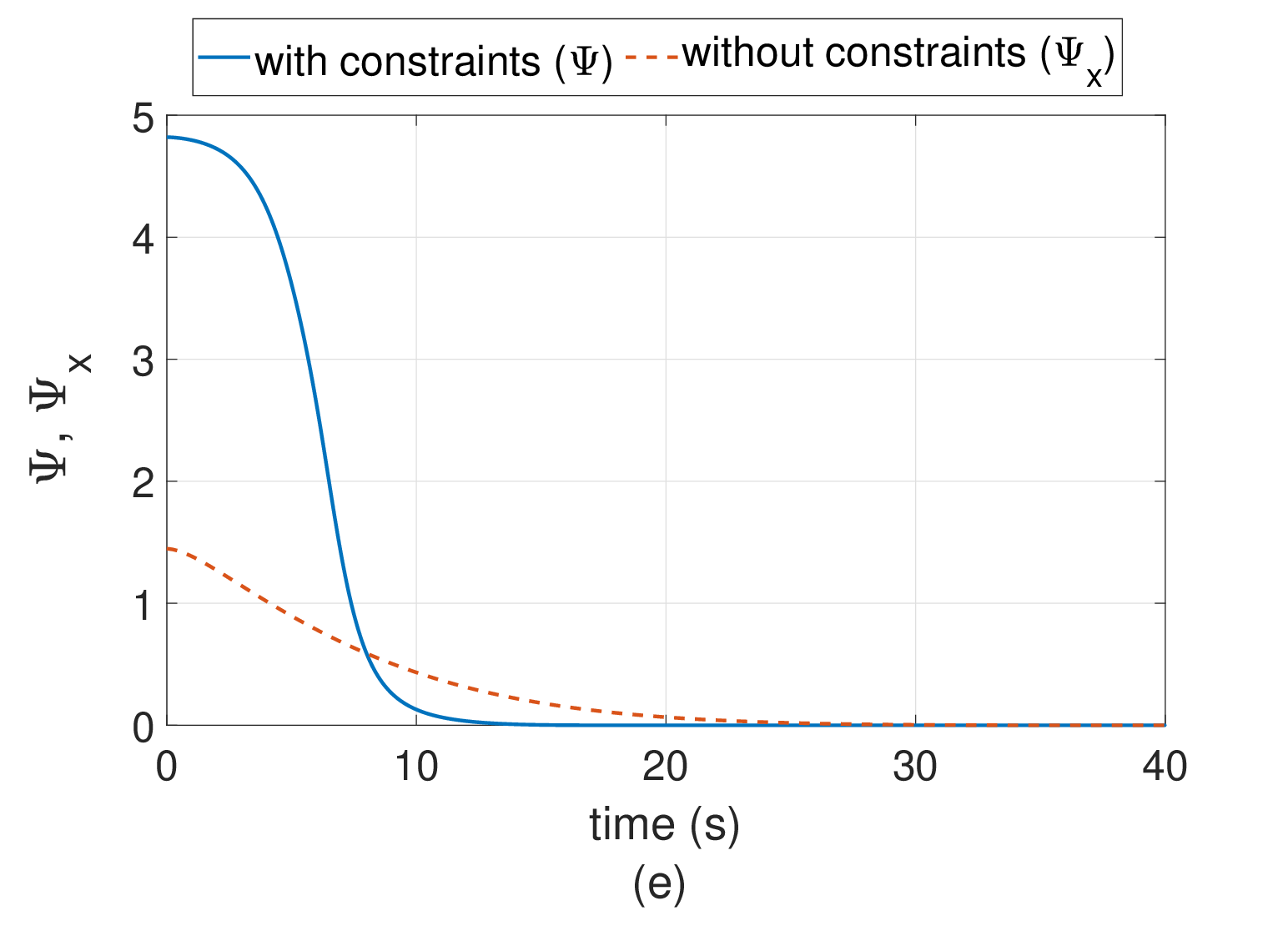}\includegraphics[width=0.7\columnwidth,totalheight=0.20\paperwidth]{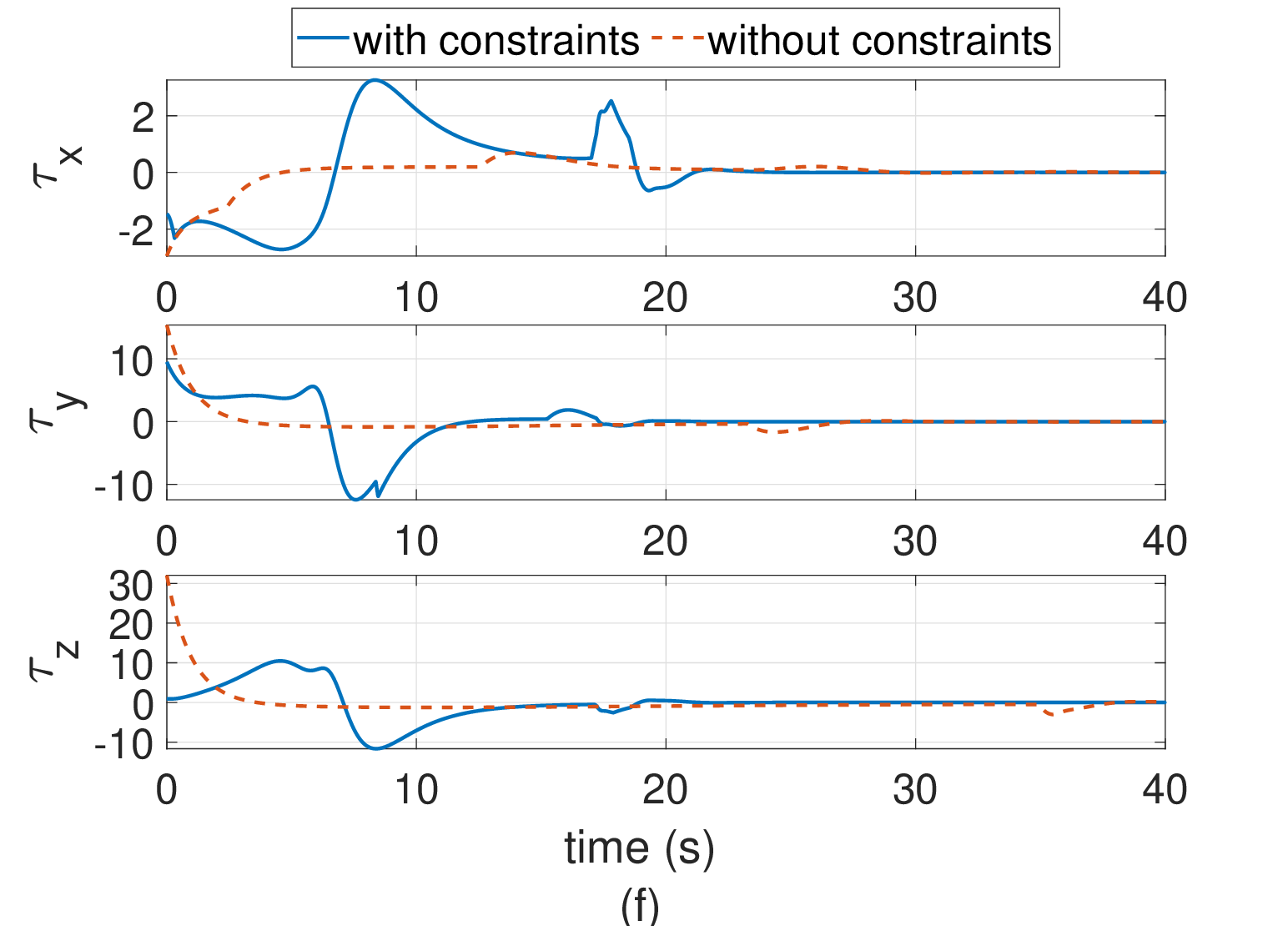}
         \includegraphics[width=\textwidth]{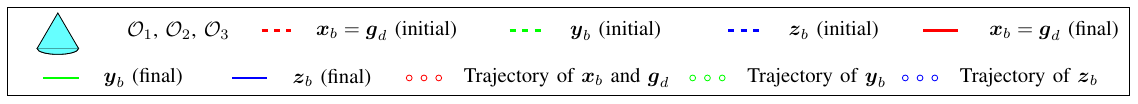}
	\end{center}
	\caption{Simulation results for Example~2: 
(a) and (c) Attitude orientation trajectories of the body-fixed frame $\mathcal{F}_{\mathcal{B}}=\{\boldsymbol{x}_b,\boldsymbol{y}_b,\boldsymbol{z}_b\}$ and the boresight vector $\boldsymbol{g}_d$, 
(b) Time response of $\mathrm{tr}(\boldsymbol{R}_e)$, 
(d) Time response of the angular velocity error vector $\boldsymbol{e}_{\omega}\coloneq[\,e_{\omega_x}\; e_{\omega_y}\; e_{\omega_z}\,]^{\top}$ ($\mathrm{rad/s}$), 
(e) Time responses of $\Psi$ and $\Psi_x$, 
(f) Time response of the control input vector $\boldsymbol{\tau}\coloneq[\,\tau_x\; \tau_y\; \tau_z\,]^{\top}$ ($\mathrm{N \cdot m}$).}
\label{fig:2}
\end{figure*}

\noindent\textbf{Example 2:} In this example, three forbidden zones are prescribed. The unit vectors
defining the corresponding forbidden directions are selected as
$\boldsymbol{y}'_{1} = [\begin{array}{@{}ccc@{}}
-0.7576 & -0.5527 & 0.3473\end{array}]^{\top}$,
$\boldsymbol{y}'_{2} = [\begin{array}{@{}ccc@{}}
-0.4134 & -0.5731 & 0.7076\end{array}]^{\top}$, and
$\boldsymbol{y}'_{3} = [\begin{array}{@{}ccc@{}}
-0.2603 & -0.6040 & -0.7533\end{array}]^{\top}$,
all expressed in the $\mathcal{F}_{\mathcal{N}}$ frame. The minimum allowable angle between $\boldsymbol{R}\boldsymbol{g}_d$ and $\boldsymbol{y}'_i$ are chosen as $\theta_i=20^\circ$ for all $i\in\{1,2,3\}$. The initial attitude and angular velocity of the spacecraft are chosen as $\boldsymbol{R}(0)=\operatorname{exp}\left(3.1180\boldsymbol{e}^\times\right)$ and $\boldsymbol{\omega}(0)=[\begin{array}{@{}ccc@{}}
0.01 & 0.01 & 0.01\end{array}]^{\top}$ ($\mathrm{rad/s}$), respectively, where $\boldsymbol{e}=[\begin{array}{@{}ccc@{}}
0 & -0.4155 & -0.8942\end{array}]^{\top}$. The desired attitude and angular velocity are chosen as $\boldsymbol{R}_d=\operatorname{exp}\left(0.5\boldsymbol{e}^\times\right)$ and $\boldsymbol{\omega}_d=[\begin{array}{@{}ccc@{}}
0 & 0 & 0\end{array}]^{\top}$, respectively. Figure \ref{fig:2}(c) shows the attitude trajectory for the constraint-free case, while Fig. \ref{fig:2}(a) presents the corresponding motion when attitude constraints are enforced; the shaded cones denote the forbidden pointing regions. As discussed in Example 1, the unconstrained motion follows the shortest rotation path and passes through regions inside or near the exclusion cones. In contrast, the proposed constrained controller modifies the trajectory so that it remains outside the forbidden regions in this example as well, despite the presence of multiple attitude-pointing constraints. As a result, the constrained path becomes slightly longer and more curved, yet the spacecraft still reaches the desired final attitude, showing that the constraints reshape only the transient motion without degrading the final tracking performance.

Figures~\ref{fig:2}(b) and \ref{fig:2}(e) show the evolution of the trace of the attitude error matrix and the configuration error functions for the constraint-free and constrained cases. In both cases, $\mathrm{tr}(\boldsymbol{R}_e)$ converges to $3$, confirming that the desired attitude is achieved. However, the constrained
trajectory exhibits a brief plateau as the controller detours around the forbidden pointing regions, whereas the unconstrained case follows a smoother, direct convergence. A similar trend is observed in the error functions: the unconstrained potential $\Psi_x$ decreases monotonically along the shortest
rotation path, while the constrained potential $\Psi$ drops rapidly at first and then decays more slowly near the constraint boundary. Ultimately, both error measures converge to zero, indicating that the constraints mainly influence the transient evolution while preserving the final attitude accuracy. Figures \ref{fig:2}(d) and \ref{fig:2}(f) show that the constrained controller introduces larger transient overshoots in the angular velocity compared with the constraint-free case, although both eventually converge to zero. The corresponding control torques exhibit a similar behavior, where the constrained case experiences higher transient peaks as the controller adjusts the attitude to remain outside the forbidden regions. Despite these differences, the overall trends remain consistent with Example 1, and the final tracking performance is preserved.

\section{Conclusions}
This paper developed an intrinsic geometric fixed-time control framework for constrained spacecraft attitude tracking on $\mathrm{SO}(3)$ in the presence of multiple attitude pointing constraints and matched external disturbances. An attitude potential function was constructed on $\mathrm{SO}(3)$, and its geometric properties, including the existence of a unique nondegenerate minimum and local strong convexity, were established using Riemannian analysis. Based on its Riemannian gradient, a nonsingular fixed-time geometric sliding manifold and the corresponding constrained attitude control law were developed. The closed-loop analysis established invariance of the admissible free space and guaranteed convergence of the state to a sufficiently small neighborhood of the desired equilibrium within a prescribed fixed time while satisfying the pointing constraints. Numerical simulations validated the theoretical results and demonstrated the effectiveness of the proposed approach.

\section*{References}

\bibliographystyle{IEEEtran}
\bibliography{References}
\end{document}